\documentclass[12pt, reqno]{amsart}
\usepackage{amssymb}
\usepackage{amsmath, mathrsfs}
\allowdisplaybreaks[1]
\usepackage{graphicx,color}
\usepackage{wrapfig,framed}

\usepackage{extpfeil}

\usepackage{color}

\usepackage[height=22.7cm, width=17cm, hmarginratio={1:1}]{geometry}
\usepackage{hyperref}

\usepackage{mathtools}

\usepackage{mathdots}

\usepackage{lscape}

\usepackage{makecell}

\theoremstyle{plain}
\newtheorem{theorem}{Theorem}
\newtheorem{prop}{Proposition}
\newtheorem{lemma}{Lemma}
\newtheorem{cor}{Corollary}
\newtheorem{cnj}{Conjecture}
\newtheorem*{theorem'}{Theorem\ 1'}
\theoremstyle{definition}

\newtheorem*{remark}{Remark}

\newcommand{\beq}{\begin{equation}}
\newcommand{\eeq}{\end{equation}}
\newcommand{\nn}{\nonumber}

\newcommand{\QQ}{{\mathbb Q}}

\newcommand{\CC}{{\mathbb C}}

\newcommand{\bt}{{\bf t}}
\newcommand{\e}{\epsilon}
\newcommand{\p}{\partial}

\newcommand{\ad}{{\rm ad}}
\newcommand{\Ker}{{\rm Ker}}
\newcommand{\diag}{{\rm diag}}
\newcommand{\TT}{{\mathcal{T}}}
\newcommand{\Tr}{{\rm Tr}}

\newcommand{\A}{\mathscr{A}}
\newcommand{\B}{\mathscr{B}}
\newcommand{\C}{\mathscr{C}}
\newcommand{\D}{\mathscr{D}}
\newcommand{\LLL}{\mathscr{L}}

\title[GW invariants of $\mathbb{P}^1_{m_1,m_2}$ and TDE]
{On Gromov--Witten invariants of $\mathbb{P}^1$-orbifolds and topological difference equations}
\author{Zhengfei Huang, Di Yang}
\address{School of Mathematical Sciences, University of Science and Technology of China, Hefei 230026, P.R.~China}
\email{huangzf063@mail.ustc.edu.cn, diyang@ustc.edu.cn}
\date{}

\begin{document}

\maketitle

\begin{abstract}
Let $(m_1, m_2)$ be a pair of positive integers. Denote by $\mathbb{P}^1$ the complex projective line, and 
 by $\mathbb{P}^1_{m_1,m_2}$ the orbifold complex projective line obtained from $\mathbb{P}^1$ by adding
$\mathbb{Z}_{m_1}$ and $\mathbb{Z}_{m_2}$ orbifold points.
In this paper we introduce a matrix linear difference equation, prove existence and uniqueness of its formal 
Puiseux-series solutions, 
 and use them to give conjectural formulas for $k$-point ($k\ge2$) functions of Gromov--Witten invariants of 
$\mathbb{P}^1_{m_1,m_2}$. Explicit expressions of the unique solutions are also obtained. 
We carry out concrete computations of the first few invariants by using the conjectural formulas.
For the case when one of $m_1,m_2$ equals~1, we prove validity of the conjectural formulas.
\end{abstract}

\tableofcontents

\section{Introduction}\label{intro}
Let $(m_1, m_2)$ be a pair of positive integers, and $\mathbb{P}^1$ the complex projective line. 
Denote by $\mathbb{P}^1_{m_1,m_2}$ the orbifold complex projective line obtained from $\mathbb{P}^1$ by adding
$\mathbb{Z}_{m_1}$ and $\mathbb{Z}_{m_2}$ orbifold points. 
In this paper we will propose a conjectural formula for cetain $k$-point generating series  
of the Gromov--Witten (GW) invariants of $\mathbb{P}^1_{m_1,m_2}$. 

In order to state the conjectural formula we first recall some terminologies about GW invariants of $\mathbb{P}^1_{m_1,m_2}$.
Recall that the orbifold cohomology of $\mathbb{P}^1_{m_1,m_2}$ is given by
$$
H_{\rm orb}(\mathbb{P}^1_{m_1,m_2})=H(I\mathbb{P}^1_{m_1,m_2})=H^0(\mathbb{P}^1_{m_1,m_2})\oplus H^2(\mathbb{P}^1_{m_1,m_2})\oplus\bigoplus_{i=1}^2\bigoplus_{j=1}^{m_1-1}H^0(B\mu_{m_i}(j)),
$$
where $I\mathbb{P}^1_{m_1,m_2}$ is the inertia orbifold of $\mathbb{P}^1_{m_1,m_2}$ and 
 $B\mu_{m_i}(j)\cong B\mu_{m_i}$ is the classifying stack of the group of $m_i$th roots of unity. 
The orbifold cohomology $H_{\rm orb}(\mathbb{P}^1_{m_1,m_2})$ carries the orbifold Poincar\'e paring 
$\langle\,,\rangle^{\mathbb{P}^1_{m_1,m_2}}$, which is non-degenerate. 
Fix a basis $(\phi_{a})_{a=0,\dots,l-1}$ of $H_{\rm orb}(\mathbb{P}^1_{m_1,m_2})$, homogenous with respect to the orbifold degree, 
as follows: $\phi_0=1\in H^0(\mathbb{P}^1_{m_1,m_2})$, $\phi_{m_1}={\rm [pt]} \in H^2(\mathbb{P}^1_{m_1,m_2})$, 
$\phi_a=1 \in H^0(B\mu_{m_1}(a))$ for $a=1,\dots,m_1-1$, and 
$\phi_a=1 \in H^0(B\mu_{m_2}(l-a))$ for $a=m_1+1,\dots,l-1$. 
Here and below,
$ l:=m_1+m_2$. 
The products between elements in this basis under $\langle\,,\rangle^{\mathbb{P}^1_{m_1,m_2}}$ satisfy that 
\begin{align}
&\langle\phi_0,\phi_{m_1}\rangle^{\mathbb{P}^1_{m_1,m_2}}=\langle\phi_{m_1}, \phi_0\rangle^{\mathbb{P}^1_{m_1,m_2}}=1,
\quad \langle\phi_a,\phi_{m_1-a}\rangle^{\mathbb{P}^1_{m_1,m_2}}=\frac1{m_1} \;\;(a=1,\dots,m_1-1),\\ 
&
\langle\phi_a,\phi_{l+m_1-a}\rangle^{\mathbb{P}^1_{m_1,m_2}}=\frac1{m_2}\;\;(a=m_1+1,\dots,l-1),
\end{align}
and vanish otherwise.
The orbifold degree of $\phi_a$, denoted as $2q_a$, is given by
\beq
q_a= \left\{\begin{array}{ll} 
\frac{a}{m_1}, & a=0,\dots,m_1, \\
\\
 \frac{l-a}{m_2}, & a=m_1+1,\dots,l-1.
\end{array}\right. 
\eeq
For more details about the orbifold cohomology of $\mathbb{P}^1_{m_1,m_2}$ see~\cite{AGV, AGV2, CR, CR2, MT}.

Let $\overline{\mathcal{M}}_{g,k}(\mathbb{P}^1_{m_1,m_2},d)$ be the moduli stack of orbifold stable maps of degree $d$ 
 from algebraic curves of genus $g$ with $k$ distinct marked points
 to $\mathbb{P}^1_{m_1,m_2}$.
Let $\mathcal{L}_i$ be the $i$th tautological line bundle on $\overline{\mathcal{M}}_{g,k}(\mathbb{P}^1_{m_1,m_2},d)$, and 
 $\psi_i:= c_1(\mathcal{L}_i),$ $i=1,\dots,k$. 
Denote by ${\rm ev}_i :   \overline{\mathcal{M}}_{g,k}(\mathbb{P}^1_{m_1,m_2},d)\rightarrow I\mathbb{P}^1_{m_1,m_2}$ the $i$th evaluation map.
The genus $g$ and degree $d$ GW invariants
of $\mathbb{P}^1_{m_1,m_2}$ are integrals of the form
\beq\label{corr-gd}
\int_{\bigl[\overline{\mathcal{M}}_{g,k}(\mathbb{P}^1_{m_1,m_2}, \,d)\bigr]^{\rm virt}} 
{\rm ev}_1^*(\phi_{a_1})  \cdots  {\rm ev}_k^*(\phi_{a_k}) \, \psi_1^{i_1} \cdots \psi_k^{i_k} \; =: \;  
\langle \tau_{i_1}(\phi_{a_1}) \cdots \tau_{i_k}(\phi_{a_k})  \rangle_{g,d}  \,.
\eeq
Here, $a_1,\dots,a_k\in\{0,\dots,l-1\},\, i_1,\dots,i_k\geq 0$
and $\bigl[\, \overline{\mathcal{M}}_{g,k}(\mathbb{P}^1_{m_1,m_2},d)\bigr]^{\rm virt}$ denotes the 
virtual fundamental class \cite{AGV2, KM}.  
These integrals vanish unless 
the degree--dimension matching holds:
\beq
2g - 2 + \frac{d}{\rho} + k = \sum_{\ell=1}^k i_\ell + \sum_{\ell=1}^k q_{a_\ell},
\eeq
where $\rho:= \frac{m_1 m_2}{m_1+m_2}$.
Clearly, $l=m_1+m_2$ is the dimension of the corresponding Frobenius manifold~\cite{Du2, DZ, MST, MT, Rossi}, 
and $\frac1\rho=\frac1{m_1}+\frac1{m_2}$ is the orbifold Euler characteristic of $\mathbb{P}^1_{m_1,m_2}$.

For $k\ge1$ and $a_1,\dots,a_k=0,\dots,l-1$, define the {\it $k$-point functions of GW invariants of $\mathbb{P}^1_{m_1,m_2}$} by
\beq
F_{a_1,\dots,a_k}(\lambda_1,\dots,\lambda_k;Q;\epsilon):=
\sum_{i_1,\dots,i_k\ge0}\prod_{j=1}^k\frac{ Q^{(1-q_{a_j})\rho}\e^{q_{a_j}}q_{a_j,i_j}}{\lambda_j^{i_j+q_{a_j}+1}}\sum_{g,\,d\ge0}\e^{2g-2}Q^d\langle \tau_{i_1}(\phi_{a_1}) \cdots \tau_{i_k}(\phi_{a_k})  \rangle_{g,d} ,
\eeq
where 
$q_{a,i}:=(q_a)_{i+1}(m_1\delta_{a<m_1}+\delta_{a,m_1}+m_2\delta_{a>m_1})$,
with $(q_a)_{m}$ being the increasing Pochhammer symbol, i.e., 
$(q_a)_{m}:=q_a(q_a+1)\cdots(q_a+m-1)$.

In studying GW 
invariants of~$\mathbb{P}^1$, the Toda lattice hierarchy 
and the corresponding topological recursion, 
the following linear difference equation was introduced~\cite{DYZ, Marchal} (cf.~\cite{DY1}):
\beq\label{TEp1}
M(z-1, s)\begin{pmatrix}z-\frac12&-s\\s&0\end{pmatrix} = \begin{pmatrix}z-\frac12&-s\\s&0\end{pmatrix} M(z, s),
\eeq
which is called in~\cite{DYZ} the {\it topological difference equation}. It was proved in~\cite{DYZ} (cf.~\cite{DY2, Marchal}) that 
there exists a unique formal solution of equation~\eqref{TEp1} 
satisfying a certain initial condition, and that this unique solution has the following explicit expression:
\beq\label{MApril8}
M(z, s)=\begin{pmatrix}1+\alpha&P_1-P_2\\ P_1+P_2&-\alpha\end{pmatrix},
\eeq
where $\alpha=\alpha(z, s)$, $P_1=P_1(z,s)$, $P_2=P_2(z,s)\in\QQ[s][[z^{-1}]]$ are given by
\begin{align}
&\alpha(z,s)=2\sum_{j=0}^\infty\frac1{z^{2j+2}}\sum_{i=0}^js^{2i+2}\frac{1}{i!(i+1)!}\sum_{\ell=0}^i(-1)^\ell(i-\ell+\frac12)^{2j+1}\binom{2i+1}{\ell},\\
&P_1(z,s)=\sum_{j=0}^\infty\frac1{z^{2j+1}}\sum_{i=0}^js^{2i+1}\frac{1}{i!^2}\sum_{\ell=0}^i(-1)^\ell(i-\ell+\frac12)^{2j}\biggl(\binom{2i}{\ell}-\binom{2i}{\ell-1}\biggr), \label{betaexp}\\
&P_2(z,s)=-\frac12\sum_{j=0}^\infty\frac1{z^{2j+2}}\sum_{i=0}^j s^{2i+1}\frac{2i+1}{i!^2}\sum_{\ell=0}^i(-1)^\ell(i-\ell+\frac12)^{2j}\biggl(\binom{2i}{\ell}-\binom{2i}{\ell-1}\biggr). \label{gammexp}
\end{align}
Moreover, the $k$-point function with $k\ge2$ has the expression:
\begin{align}
&F(\lambda_1,\dots,\lambda_k;Q;\e)=-\frac1 k\sum_{\sigma\in S_k}\frac{\Tr \, M(\frac{\lambda_{\sigma_1}}{\e},\frac{Q^{1/2}}\e)\cdots M(\frac{\lambda_{\sigma_k}}{\e},\frac{Q^{1/2}}\e)}{\prod_{i=1}^k (\lambda_{\sigma(i)}-\lambda_{\sigma(i+1)})} 
- \delta_{k,2}\frac1{(\lambda_1-\lambda_2)^2}. \label{Fkp1}
\end{align}
Here, the short notation $F$ means $F_{1,\dots,1}$, $S_k$ denotes the symmetric group and it is understood that $\sigma(k+1)=\sigma(1)$.
Identity~\eqref{Fkp1} with $M$ given by~\eqref{MApril8}--\eqref{gammexp} was 
 conjectured in~\cite{DY2} and proved in~\cite{DYZ, Marchal}.

It was suggested in~\cite{DY2} that the above formulas~\eqref{MApril8}--\eqref{Fkp1} could be generalized 
to GW invariants of $\mathbb{P}^1$-orbifolds~\cite{IST, MST, Rossi}. 
In this paper we will achieve such a generalization (see Conjecture~\ref{cnj1} and Theorem~\ref{thmexpM} below) for the $A$-series (cf.~\cite{DZ, MST, Rossi}). 

We call the following linear equation
\beq\label{TE}
M(z-1, s)W(z, s)=W(z, s)M(z,s)
\eeq
for an $l\times l$ matrix-valued function $M(z, s)$ 
 the {\it topological difference equation of $(m_1,m_2)$-type}, for short the {\it topological difference equation (TDE)},
 where 
\beq\label{Q}
W(z , s) =  (z-\frac12) e_{1,m_1} - s e_{1,l} + s\sum_{i=2}^le_{i,i-1}.
\eeq
Here $e_{i,j}$ is the matrix (of according size, here $l\times l$) with the $(i,j)$-entry being~$1$ and others $0$.
For the case when $m_1=m_2=1$,  it is easy to see that equation~\eqref{TE} indeeds coincides with~\eqref{TEp1}. 
The motivation of the above definition~\eqref{TE} also comes from the topological differential equations introduced in~\cite{BDY2} and from the 
matrix-resolvents obtained in~\cite{FYZ} for the bigraded Toda hierarchy of $(m_1,1)$-type. 

Introduce some notations: 
\beq 
K_a:=
\left\{\begin{array}{ll}
\sum_{j=1}^{a} e_{j,m_1-a+j}, &a=1,\dots,m_1, \\
\\
-\sum_{j=1}^{l-a} e_{a+j,m_1+j},&a=m_1+1,\dots,l-1,
\end{array}\right.   
\eeq
As a generalization of \cite[Proposition~1]{DYZ} (see also~\cite{BDY2}), we will prove in Section~\ref{proof-thm1} the following
\begin{theorem}\label{prop1}
There exist unique formal solutions $M_a(z,s)$ in $z^{1-q_a}{\rm Mat}(l\times l,\mathbb{C}(s)((z^{-1})))$, $a=1,\dots, l-1$, to the TDE~\eqref{TE} such that  
\beq\label{form}
M_a(z,s) = z^{1-q_a}(K_a + O(z^{-1}) ).
\eeq
\end{theorem}

Let $M_a(z,s), a=1,\dots,l-1$, be the solutions to~\eqref{TE} given in Theorem~\ref{prop1}.
We propose in this paper the following conjecture. 
\begin{cnj}\label{cnj1}
For $k\ge2$ and $a_1,\dots,a_k=1,\dots,l-1$, the  
$k$-point functions of GW invariants of $\mathbb{P}^1_{m_1,m_2}$
have the following expressions:
\begin{align}\label{npoint}
&F_{a_1,\dots,a_k}(\lambda_1,\dots,\lambda_k;Q;\epsilon)
=-\sum_{\sigma\in S_k/C_k}
\frac{\Tr \, M_{a_{\sigma(1)}} \bigl(\frac{\lambda_{\sigma(1)}}{\e},\frac { Q^\rho}\e\bigr)\cdots 
M_{a_{\sigma(k)}}\bigl(\frac{\lambda_{\sigma(k)}}{\e},\frac{ Q^\rho}\e\bigr)}{\prod_{i=1}^k (\lambda_{\sigma(i)}-\lambda_{\sigma(i+1)})}\notag
\\
&-\delta_{k,2}
\frac{\delta_{a_1+a_2, m_1}
(a_1\lambda_1+a_2\lambda_2)+\delta_{a_1,m_1} \delta_{a_2,m_1}m_1\lambda_1\lambda_2 \e
+\delta_{a_1+a_2, m_1+l}((l-a_1)\lambda_1+(l-a_2)\lambda_2)}
{ \lambda_1^{q_{a_1}}\lambda_2^{q_{a_2}} (\lambda_1-\lambda_2)^2 \e },
\end{align}
where we recall that $\rho=m_1m_2/(m_1+m_2)$.
\end{cnj}
We note that Conjecture~\ref{cnj1} was proved in~\cite{DYZ} for the case when $m_1=m_2=1$.

It is not difficult to deduce from Conjecture~\ref{cnj1} the following corollary (cf.~\cite{DYZ}).
\begin{cor}[*]
The 1-point function satisfies
\beq\label{1point}
\e\p_\lambda (F_{a_1}(\lambda;Q;\e))
\,=\, \frac{\delta_{a, \,m_1}}{\lambda} \,-\, \frac1{ Q^\rho } \, \frac{\e\p_\lambda}{1-e^{-\e\partial_\lambda}}\Bigl(M_a\Bigl(\frac{\lambda}{\e}-1,\frac{ Q^\rho }\e\Bigr)_{m_1+1,1}\Bigr).
\eeq
\end{cor}
\noindent Here and below a statement marked with “$\,^*\,$” means that it is a consequence of Conjecture~\ref{cnj1}.

Denote by $B_m(\ell,x)$ the generalized Bernoulli polynomials, which are defined by 
$$
\bigl(\frac{t}{e^t-1}\bigr)^\ell e^{xt}=:\sum_{m\ge0}B_m(\ell,x)\frac{t^m}{m!},
$$
with $B_{m}(1,0)=B_{m}$ being the Bernoulli numbers.
We will prove in Section~\ref{explict} the following theorem. 
\begin{theorem}\label{thmexpM}
The entries of the matrix $M_{a}(z,s)$ have the explicit expressions:
\beq\label{expression}
(M_{a}(z,s))_{ij}=\left\{\begin{array}{ll}
g_a(z,i,j),&j\le m_1,\\
-g_a(z,i,j-m_1-m_2),&j> m_1,\\
\end{array}\right.
\eeq 
where, for $a=1,\dots,m_1$, 
\begin{align}
&g_a(z,i,j)= 
z^{-\frac{a}{m_1}}
\sum_{\ell_1\ge-1}\frac{m_1^{\ell_1}}{z^{\ell_1}}
\sum_{\ell_2=-1}^{\ell_1}
\delta_{m_2 | (m_1\ell_2+a+j-i)} 
\frac{s^{\frac{l\ell_2+a+j-i}{m_2}+1}}{m_1^{\ell_2}m_2^{\frac{m_1\ell_2+a+j-i}{m_2}}} \nn\\
&\times\sum_{\ell_3=0}^{\frac{m_1\ell_2+a+j-i}{m_2}}
\frac{(-1)^{\ell_3} \binom{-\ell_2-\frac{a}{m_1}}{\ell_1-\ell_2}}{\ell_3!(\frac{m_1\ell_2+a+j-i}{m_2}-\ell_3)!}
B_{\ell_1-\ell_2}\Bigl(1-\ell_2-\frac{a}{m_1},\frac{i-\frac12-a+m_2\ell_3}{m_1}-\ell_2\Bigr), \label{expression1}
\end{align}
and, for $a=m_1+1,\dots,l-1$,
\begin{align}
& g_a(z,i,j)= 
z^{-\frac{l-a}{m_2}}\sum_{\ell_1\ge-1}\frac{m_2^{\ell_1}}{z^{\ell_1}}\sum_{\ell_2=-1}^{\ell_1}
\delta_{m_1| (m_2\ell_2+i-j-a)}
\frac{s^{\frac{l\ell_2+i-j-a}{m_1}+1}}{m_1^{\frac{m_2\ell_2+i-j-a}{m_1}}m_2^{\ell_2}} \nn\\
&\times\sum_{\ell_3=0}^{\frac{m_2\ell_2+i-j-a}{m_1}}
\frac{(-1)^{\ell_3}\binom{-\ell_2-\frac{l-a}{m_2}}{\ell_1-\ell_2}}{\ell_3!(\frac{m_2\ell_2+i-j-a}{m_1}-\ell_3)!}
B_{\ell_1-\ell_2}\Bigl(1-\ell_2-\frac{l-a}{m_2},\frac{j-\frac12+a-m_1+m_1\ell_3}{m_2}-\ell_2\Bigr).\label{expression2}
\end{align} 
\end{theorem}

For the case when $m_1=m_2=1$, 
Theorem~\ref{thmexpM} was proved in~\cite{DYZ} (our proof here will be slightly different from~\cite{DYZ}), and 
according to~\cite{DY1, DYZ} or~\cite{Marchal} this theorem 
%the explicit expression of $M_1(z,s)$ given by 
 leads to a proof of the conjecture in~\cite{DY2}.

Consider the following generating series~$\mathcal{F}$ of GW invariants of $\mathbb{P}^1_{m_1,m_2}$:
\beq\label{22}
\mathcal{F}=\mathcal{F}(\mathbf{T};Q;\e):=
\sum_{k\ge0}\frac1{k!}
\sum_{\substack{0\le a_1,\dots,a_k\le l-1\\i_1,\dots,i_k\ge0}}
T^{a_1}_{i_1}\dots T^{a_k}_{i_k}
\sum_{g\ge0}\sum_{d\ge0}
\e^{2g-2}
\langle\tau_{i_1}(\phi_{\alpha_1})\dots\tau_{i_k}(\phi_{a_k})\rangle_{g,d},
\eeq
where $\mathbf{T}=(T^a_j)_{0\le a\le l-1,\,j\ge0}$. 
This generating series is often called the free energy, which satisfies the following string equation 
\beq\label{string}
\sum_{a=0}^{m_1}\sum_{j\ge1}T^a_j\frac{\partial \mathcal{F}}{\partial T^a_{j-1}}+\frac{1}{2\e^2}\biggl(\sum_{a=0}^{m_1}T^a_0T^{m_1-a}_0+\sum_{a=m_1+1}^{l-1}T^a_0T^{l+m_1-a}_0\biggr)=\frac{\partial \mathcal{F}}{\partial T^{0}_0}.
\eeq
The exponential $e^{\mathcal{F}}=:Z$ is called the partition function of GW invariants of $\mathbb{P}^1_{m_1,m_2}$.

Let us say more about the motivation of Conjecture~\ref{cnj1}, and give a proof of some part of it. 
In~\cite{MT} Milanov--Tseng constructed certain integrable systems written as Hirota type bilinear equations, 
and proved that the partition function~$Z$ 
satisfies these equations. Milanov--Tseng also conjectured~\cite{MT} that $Z$ is a particular tau-function 
for the extended bigraded Toda hierarchy~\cite{Carlet}.
In \cite{CL} Carlet--van de Leur proved the conjecture of Milanov--Tseng
(see e.g. \cite{Du, Du2, Du3, DZ2, EHY, Getzler, OP, OP2, Zhang} for the case when $m_1=m_2=1$). 

Let 
\beq\label{definitionL}
L:= \TT^{m_1} + u_{m_1-1} \TT^{m_1-1} + \dots +u_1 \TT +u_0 + u_{-1} \TT^{-1} + \dots 
+u_{-(m_2-1)} \TT^{-(m_2-1)} + u_{-m_2} \TT^{-m_2}
\eeq
be the Lax operator for the extended bigraded Toda hierarchy, where $\TT=e^{\e \p_x}$ with $x=T^0_0$.
Similar to \cite{DY1, DY2, DYZ}, using equation~\eqref{string} and the definition of tau-function~\cite{Carlet}, 
we find that the initial data of the solution corresponding to GW invariants of $\mathbb{P}^1_{m_1,m_2}$  is given by 
\begin{align}
&u_0(x,\mathbf{0};\e)=x+\frac{\e}{2},\quad u_{-m_2}(x,\mathbf{0};\e)=1,
\label{BTH-initial1}\\
&u_a(x,\mathbf{0};\e)=0, \quad a\in\{1,\dots,m_1-1\}\cup\{-m_2+1,\dots,-1\}.\label{BTH-initial2}
\end{align}
In \cite{BDY1} Bertola, Dubrovin and the second author of the present paper 
introduced the matrix-resolvent method of calculating logarithmic derivatives of tau-functions for the KdV hierarchy; this method was extended to the Toda lattice hierarchy in \cite{DY1}.  
For the case when $m_1=m_2=1$,  Conjecture~\ref{cnj1} was proved in \cite{DYZ} using this method.
By extending the matrix-resolvent method to the bigraded Toda hierarchy one should be able to prove Conjecture~\ref{cnj1}. 
For the case when $m_2=1$, the extension has been achieved in~\cite{FYZ}. 
This together with a certain symmetry structure given in Corollary~\ref{cor-symm2} (see below in Section~\ref{proof-thm1}) allows us to prove the following theorem, which gives main 
evidence for the validity of Conjecture~\ref{cnj1}.
\begin{theorem}\label{thm-GW}
When one of $m_1, m_2$ is $1$, Conjecture~\ref{cnj1} holds.
\end{theorem}

The rest of the paper is organized as follows. 
In Section~\ref{proof-thm1}, we prove Theorem~\ref{prop1} and give more properties of the unique solutions given in the 
theorem. In Section~\ref{explict}, we give the explicit expressions for the unique solutions. 
In Section~\ref{polygon}, based on the Conjecture \ref{cnj1}, 
we employ an algorithm designed in~\cite{DY1, DY2} to give concrete computations for some of the GW invariants.
% and give examples for the case $(2,1)$, $(3,1)$ and $(2,2)$. 
In Section~\ref{conclusion}, we prove Theorem \ref{thm-GW}.

\medskip

\noindent {\bf Acknowledgements}  We thank Alexander Alexandrov and Hua-Zhong Ke for helpful discussions. D.Y. is partially supported by NSFC No.~12371254 and 
CAS No.~YSBR-032.

\section{Particular formal solutions to the TDE}\label{proof-thm1}
The goal of this section is to prove Theorem~\ref{prop1}.

We first introduce some notations. Denote 
\beq
G(z,s)=W(z+\frac12,s).
\eeq
Denote $\LLL={\rm Mat}(l\times l,\CC(s)((z^{-1})))$, 
$\A={\rm Mat}(m_1\times m_1,\CC(s)((z^{-1})))$,
$\B={\rm Mat}(m_1\times m_2,\CC(s)((z^{-1})))$,
$\C={\rm Mat}(m_2\times m_1,\CC(s)((z^{-1})))$,
and 
$\D={\rm Mat}(m_2\times m_2,\CC(s)((z^{-1})))$, where we recall that $l=m_1+m_2$. An element in $\LLL$ will often be written as 
$\begin{pmatrix}
A&B\\C&D
\end{pmatrix}$, where $A\in \A$, $B\in\B$, $C\in\C$, $D\in\D$.
In particular, we write the matrix $G(z,s)$ as 
\beq
G(z,s)=\begin{pmatrix} G_1(z,s) & G_2(z,s) \\ G_3(z,s) & G_4(z,s)\end{pmatrix},
\eeq
where $G_1(z,s)=z e_{1,m_1}+s \sum_{i=2}^{m_1} e_{i,i-1} \in \A$, $G_2(z,s)=-se_{1,m_2} \in \B$, $G_3(z,s)=s_{1,m_1} \in \C$, 
and $G_4(z,s)=s\sum_{i=2}^{m_2}e_{i,i-1} \in \D$.

To prove Theorem~\ref{prop1}, we will actually prove the following equivalent version. 

\noindent {\bf Theorem~$1'$.} {\it 
There exist unique formal solutions $Y_a(z,s)$ in $z^{1-q_a}\cdot \LLL$, $a=1,\dots,l-1$, to the equation
\beq\label{TE'}
Y(z-1,s)G(z,s)=G(z,s)Y(z,s)
\eeq
such that}
\beq\label{form'}
Y_a(z,s)=z^{1-q_a}(K_a+O(z^{-1})).
\eeq

Before proving Theorem~$1'$, we do some preparations.

For any $m\ge1$, define 
an inner product $\langle \,, \rangle_m$ on 
${\rm Mat}(m\times m,\CC(s)((z^{-1})))$
 by 
$$
\langle M_1, M_2\rangle_m:=\Tr \, M_1M_2.
$$
For simplifying the notations, we denote $G_i=G_i(z,s)$, $i=1,\dots,4$.

The following lemma can be found for example in~\cite{DS}.
\begin{lemma}[\cite{DS}]\label{orth-A}
We have
\begin{align}
&{\rm Im}\ \ad_{G_{1}}=(\Ker\ \ad_{G_{1}})^\perp,\\
&\Ker\,\ad_{G_1}={\rm Span}_{\CC(s)((z^{-1}))} \bigl\{G_1^{j} \,| \, j=0,\dots,m_1-1\bigr\},\\
&\A=\Ker\ \ad_{G_1}\oplus{\rm Im}\ \ad_{G_1},\label{directsum}
\end{align}
where the orthogonality is with respect to $\langle\,,\rangle_{m_1}$.
\end{lemma}

Similar to the above lemma we will prove the following 
\begin{lemma}\label{orth-D}
We have
\begin{align}
&\Ker\,\ad_{G_4}={\rm Span}_{\CC(s)((z^{-1}))}\{G_4^{j} \,|\, j=0,\dots,m_2-1\}, \label{G4Kernel}\\
&{\rm Im} \, \ad_{G_4}=(\Ker\, \ad_{G_4})^\perp, \label{G4Imortho}
\end{align}
where the orthogonality is with respect to $\langle \,,\rangle_{m_2}$.
\end{lemma}
\begin{proof}
For $M_0(z,s)=(d_{i,j}(z,s))_{i,j=1,\dots,m_2}\in\Ker\,\ad_{G_4}$, we have
\beq\label{ker4}
[G_4,M_0(z,s)]_{i,j}=d_{i-1,j}(z,s)-d_{i,j+1}(z,s)=0,\quad i,j=1,\dots,m_2.
\eeq
Here it is understood that $d_{i,m_2+1}(z,s)$ and $d_{0,j}(z,s)$ are $0$. By solving equation \eqref{ker4} we get
\begin{align*}
&d_{i,j}(z,s)=0,\quad 1\le i<j\le m_2,\\
&d_{i,j}(z,s)=d_{i-j+1, \,1}(z,s),\quad 1\le j\le i\le m_2,
\end{align*}
where $d_{i,1}(z,s)\in\CC(s)((z^{-1}))$, $i=1,\dots,m_2$, are free.
From this it can follow that $(G_4^j)_{j=0,\dots,m_2-1}$ form a basis of $\Ker\,\ad_{G_4}$, namely, equation~\eqref{G4Kernel} is proved.

For each element $M_1(z,s)\in {\rm Im}\, \ad_{G_4(z,s)}$, there exists $M_2(z,s)\in\D$ 
such that $M_1(z,s)=[G_4(z,s),M_2(z,s)]$. Then for any $M_3(z,s)\in\Ker\, \ad_{G_4(z,s)}$, 
$$
\langle M_3(z,s),M_1(z,s)\rangle=\Tr \,\bigl([M_3(z,s),G_4(z,s)]M_2(z,s)\bigr)=0.
$$ 
So ${\rm Im}\, \ad_{G_4(z,s)}\subset(\Ker\, \ad_{G_4(z,s)})^\perp$.
On another hand,  
$$
\dim_{\CC(s)((z^{-1}))} {\rm Im}\, \ad_{G_4(z,s)}=m_2^2-\dim_{\CC(s)((z^{-1}))} \Ker\, \ad_{G_4(z,s)}
=\dim_{\CC(s)((z^{-1}))}(\Ker\, \ad_{G_4(z,s)})^{\perp}.
$$
The equality~\eqref{G4Imortho} is proved.
\end{proof}

For each $A(z,s)\in\A$, write $A(z,s)=A(z,s)_{\Ker_1}+A(z,s)_{{\rm Im}_1}$ with $A(z,s)_{\Ker_1}\in\Ker\,\ad_{G_1}$, $A(z,s)_{{\rm Im}_1}\in{\rm Im}\,\ad_{G_1}$. We fix an $S\subset\D$ such that $\D=\Ker\,\ad_{G_4}\oplus S$. For each $D(z,s)\in\D$, write $D(z,s)=D(z,s)_{\Ker_2}+D(z,s)_S$, with $D(z,s)_{\Ker_2}\in\Ker\,\ad_{G_4}$ and $D(z,s)_S\in S$.

We continue to do some more preparations. 

For a block-matrix
$$
\begin{pmatrix}
A&B\\C&D
\end{pmatrix} \in \LLL,
$$
where $A\in\A$, $B\in\B$, $C\in\C$, $D\in\D$, we introduce degree assignments $\deg_{11}$ on~$\A$, $\deg_{12}$ on~$\B$, $\deg_{21}$ on~$\C$, $\deg_{22}$ on~$\D$ by
$$
\deg_{11}e_{i_1, j_1}=i_1-j_1,\quad\deg_{12}e_{i_2, j_2}=i_2-m_1,\quad\deg_{21}e_{i_3, j_3}=m_1-j_3,\quad\deg_{22}e_{i_4, j_4}=0,
$$
$$
\deg_{11}z=
\deg_{12}z=
\deg_{21}z=
\deg_{22}z=m_1.
$$
Following~\cite{BDY2}, introduce the notations
$$
{\rm gr}_{11}=m_1z\partial_z+\ad_{\rho^\vee},\quad {\rm gr}_{22}= m_1 z \p_z.
$$
Here $\rho^\vee=\diag(\frac{1-m_1}{2},\frac{3-m_1}{2},\dots,\frac{m_1-3}{2},\frac{m_1-1}{2})$.
It is known from e.g.~\cite{BDY2, Kostant} that $A$ is homogenous of degree $d$ with respect to $\deg_{11}$ if and only if 
\beq\label{grA}
{\rm gr}_{11} \, A=dA.
\eeq
Obviously, 
 $D$ is homogenous of degree $d$ with respect to $\deg_{22}$ if and only if 
\beq\label{grD}
{\rm gr}_{22} \, D=dD.
\eeq
We also denote by $\A^{\le d}$ the subspace of~$\A$ whose elements have degrees less than or equal to~$d$ with respect to $\deg_{11}$, by 
 $\B^{\le d}$ the subspace of~$\B$ whose elements have degrees less than or equal to $d$ with respect to $\deg_{12}$, and 
 notations  $\C^{\le d}$ and $\D^{\le d}$ are similarly introduced. 
 
We are ready to state and prove the following lemma.
\begin{lemma}\label{thm1''}
There exist unique formal solution $Y_a(z,s)=\begin{pmatrix}A(z,s)&B(z,s)\\C(z,s)&D(z,s)\end{pmatrix}$ in $z^{1-q_a} \cdot \LLL$, $a=1,\dots,l-1$, 
to equation~\eqref{TE'} such that
\begin{align}
z^{q_a-1}A(z,s)&-s^{1-a}z^{-1}G_1^a \; \in\; \A^{\le -m_1},\label{form''1}\\
z^{q_a-1}B(z,s)&-s^{1-a}z^{-1}\sum_{i=2m_1-a}^{2m_1-2}G_1^{2m_1-1-i}G_2G_4^{i-2m_1+a} \;\in\; \B^{\le 1-2m_1},\\
z^{q_a-1}C(z,s)&-s^{1-a}z^{-1}\sum_{i=m_1-a+1}^{m_1-1}G_4^{i-m_1+a-1}G_3G_1^{m_1-i} \;\in\; \C^{\le -m_1},\\
z^{q_a-1}D(z,s)& \;\in\; \D^{\le -m_1}
\end{align}
for $a=1,\dots,m_1$, and that 
\begin{align}
z^{q_a-1}A(z,s)&\in\A^{\le -m_1},\\
z^{q_a-1}B(z,s)&+s^{m_1-a}\sum_{i=m_1}^{2m_1-2}G_1^{m_1-i-1}G_2G_4^{a+i-2m_1}\in\B^{\le 1-2m_1},\\
z^{q_a-1}C(z,s)&+s^{m_1-a}\sum_{i=1}^{m_1-1}G_4^{a+i-m_1-1}G_3G_1^{-i}\in\C^{\le -m_1},\\
z^{q_a-1}D(z,s)&+s^{m_1-a}G_4^{a-m_1}\in\D^{\le -m_1}\label{form''2}
\end{align}
for $a=m_1+1,\dots,l-1$. 
\end{lemma}
\begin{proof}
Let us fix an $a\in\{1,\dots,l\}$. Write
\begin{align}\label{2.4}
&Y(z,s)=z^{1-q_a}\sum_{i\ge0} \begin{pmatrix}A^{[-i]}&B^{[-i]}\\C^{[-i]}&D^{[-i]}\end{pmatrix},
\end{align}
with the first few terms be determined by \eqref{form''1}--\eqref{form''2}. 
Here, for $i\ge0$, $A^{[-i]}=A^{[-i]}(z,s)\in \A$, $B^{[-i]}=B^{[-i]}(z,s)\in \B$, $C^{[-i]}=C^{[-i]}(z,s)\in \C$, $D^{[-i]}=D^{[-i]}(z,s)\in \D$ are homogeneous 
of degrees $-i$ with respective to $\deg_{11}$, $\deg_{12}$, $\deg_{21}$, $\deg_{22}$, respectively.
Obviously, $D^{[-i]}$ vanish unless $m_1|i$.

Substituting \eqref{2.4} in~\eqref{TE'} and comparing terms with equal degrees, we find that~\eqref{TE'} is 
equivalent to the following equations:
\begin{align}
&\Bigl[G_1,A^{[-i-1]}\Bigr]=0,\quad i=-1,\dots,m_1-2,\label{rec1-0}\\
&\Bigl[G_4,D^{[-i]}\Bigr]=0,\quad i=0,\dots,m_1-1,\label{rec4-0}\\
&\Bigl[G_1,A^{[-i-1]}\Bigr]=
\bigl(\widetilde A^{[-i-1]}-A^{[-i-1]}\bigr)G_1+\widetilde B^{[-i]}G_3-G_2C^{[m_1-1-i]},\quad i\ge m_1-1,\label{rec1}\\
&\Bigl[G_4,D^{[-i]}\Bigr]=\widetilde C^{[m_1-1-i]}G_2+\bigl(\widetilde D^{[-i]}-D^{[-i]}\bigr)G_4-G_3B^{[-i]},\quad i\ge m_1,\label{rec4}\\
&G_1B^{[-i-1]}=\widetilde A^{[m_1-1-i]}G_2+\widetilde B^{[-i]}G_4-G_2D^{[m_1-1-i]},\quad i\ge-1,\label{rec2}\\
&C^{[-i-1]}G_1=G_3A^{[-i]}+G_4C^{[-i]}-\widetilde D^{[-i]}G_3-(\widetilde C^{[-i-1]}-C^{[-i-1]})G_1,\quad i\ge-1,\label{rec3}
\end{align}
where 
$$
\widetilde X^{[-i]}=\widetilde X^{[-i]}(z,s)
:=\sum_{j\ge0}\frac{(-1)^j}{j!}\sum_{\ell=0}^j\binom{1-q_a}{\ell} z^{-\ell}\partial_z^{j-\ell} \bigl(X^{[m_1j-i]}\bigr), \quad i\ge0,
%=A_i-(\frac{1-q_a}zA_{i-m_1}+A_{i-m_1}')+\dots
$$
with $X=A,B,C$, or $D$. 
Here and below, 
it is understood that $A^{[i]}=A^{[i]}(z,s), B^{[i]}=B^{[i]}(z,s), C^{[i]}=C^{[i]}(z,s), D^{[i]}=D^{[i]}(z,s)$ are $0$ if $i>0$.
Obviously, $\widetilde A^{[-i]}\in \A$, $\widetilde B^{[-i]}\in \B$, $\widetilde C^{[-i]}\in \C$, $\widetilde D^{[-i]}\in \D$ are homogeneous 
of degrees $-i$ with respective to $\deg_{11}$, $\deg_{12}$, $\deg_{21}$, $\deg_{22}$, respectively.
It follows that $\widetilde X^{[-i]}-X^{[-i]}$ is determined by 
$X^{[m_1-i]},X^{[2m_1-i]},\dots$, namely, it does not contain explicitly the $X^{[-i]}$-term, where $X=A,B,C$, or $D$.
Using \eqref{rec2}, \eqref{rec3} and~\eqref{rec4}, we obtain
\begin{align}
&\Bigl[G_4,D^{[-i]}\Bigr]\nn\\
&=\sum_{j=0}^{m_1}\bigl(G_4^jG_3A^{[m_1+j-i]}G_1^{-1-j}G_2
-G_3 G_1^{-1-j}\widetilde A^{[m_1+j-i]} G_2 G_4^j\bigr)\notag\\
&+\sum_{j=0}^{m_1}\bigl(G_3 G_1^{-1-j}G_2\widetilde D^{[m_1+j-i]}G_4^{j}-G_4^jD^{[m_1+j-i]}G_3G_1^{-1-j}G_2\bigr)\notag\\
&+\sum_{j=0}^{m_1}G_3G_1^{-1-j} \bigl(B^{[1+j-i]}-\widetilde B^{[i+j-i]}\bigr) G_4^{j+1}
+\sum_{j=0}^{m_1}G_4^{j+1}\bigl(C^{[m_1-i]}-\widetilde C^{[m_1-i]}\bigr) G_1^{-1-j}G_2\notag\\
&+G_4^{m_1+1}\widetilde C^{[2m_1-i]}G_1^{-1-m_1}G_2-G_1^{-1-m_1}B^{[m_1+1-i]}G_4^{m_1+1}
+ \bigl(\widetilde D^{[-i]}-D^{[-i]}\bigr)G_4,\quad i\ge m_1.\label{rec4''}
\end{align}
It is not difficult to show that the set of equations \eqref{rec1-0}--\eqref{rec3} are actually 
equivalent to equations \eqref{rec1-0}--\eqref{rec1}, \eqref{rec2}, \eqref{rec3} and~\eqref{rec4''}. 
Thus to prove the statement of the lemma it remains to show 
the existence and uniqueness for $A^{[-i]}, B^{[-i]},C^{[-i]},D^{[-i]}$ with the conditions~\eqref{form''1}--\eqref{form''2}.

To this end, we will use the mathematical induction to show the following statement: 
for all $j_0\ge m_1$, 
we have that 
$(A^{[-j]})_{\Ker_1}, (A^{[-j-m_1]})_{{\rm Im}_1}, B^{[1-m_1-j]}, C^{[-j]}, (D^{[-j]})_{\Ker_2},(D^{[-j-m_1]})_S$ for $-m_1\le j<j_0$ can be 
uniquely determined by \eqref{rec1-0}, \eqref{rec4-0}, \eqref{rec1}, \eqref{rec2}, \eqref{rec3}, \eqref{rec4''} under \eqref{form''1}--\eqref{form''2}, 
and that equation \eqref{rec1-0}, equation \eqref{rec4-0}, equation \eqref{rec1} with $i< j_0+m_1-1$, 
equation \eqref{rec2} with $i< j_0+m_1-2$, \eqref{rec3} with $i< j_0-1$ and \eqref{rec4''} with $i<j_0+m_1$ all hold, as well as that the right-hand side of \eqref{rec4''} with $i<j_0+2m_1$ belongs to ${\rm Im}\,\ad_{G_4}$.

For $j_0=m_1$, by a direct computation we find that the statement is indeed true. 

Now, assuming that the statement is true for $j_0=i_0$, we will prove it for $j_0=i_0+1$. 

First, we can solve \eqref{rec2} with $i=i_0+m_1-2$ and \eqref{rec3} with $i=i_0-1$, and from this we uniquely determine $B^{[-i_0-m_1+1]}$ and $C^{[-i_0]}$.

Second, we determine $(A^{[-i_0}])_{\Ker_1}$ and $(A^{[-i_0-m_1]})_{{\rm Im}_1}$. Equation \eqref{rec1} with $i=i_0+m_1-1$ can be written as
\beq\label{Apr3}
\Bigl[G_1,A^{[-i_0-m_1]}\Bigr]=z^{-1}\Bigl(\bigl(1-q_a-\frac{i_0}{m_1}\bigr)(A^{[-i_0]})_{\Ker_1}+\frac{1}{m_1}\bigl[\rho^\vee,(A^{[-i_0]})_{\Ker_1}\bigr]\Bigr)G_1+f_1(-i_0-m_1+1,z,s),
\eeq
where 
\begin{align}
&f_1(-i_0-m_1+1,z,s)
=\Bigl(z^{-1}(1-q_a-\frac{i_0}{m_1})(A^{[-i_0]})_{{\rm Im}_1}+\frac{1}{m_1}\bigl[\rho^\vee,(A^{[-i_0]})_{{\rm Im}_1}\bigr]\Bigr)G_1\notag\\
&+\sum_{j\ge0}\frac{(-1)^j}{j!}\sum_{\ell=0}^j\binom{1-q_a}{\ell} z^{-\ell}\partial_z^{j-\ell}\bigl(A^{[m_1(j-1)-i_0]}\bigr)G_1+\widetilde B^{[-i_0-m_1+1]}G_3-G_2C^{[-i_0]}.\label{rec1''}
\end{align}
Here we have used~\eqref{grA}. % to write $\p_z\bigl(A^{[-i_0]}\bigr)=-\frac{1}{m_1z}(i_0A^{[-i_0]}+[\rho,A^{[-i_0]}])$.
The requirement that the right-hand side of~\eqref{Apr3} belongs to ${\rm Im}\, \ad_{G_1}$ gives
$$
(A^{[-i_0]})_{\Ker_1}=\frac{s^{m_1-1}}{i_0+m_1q_a-m_1} \, \sum_{j=0}^{m_1-1}\Tr \bigl(f_1(-i_0-m_1+1,z,s)G_1^{m_1-1-j}\bigr)G_1^j.
$$
And $(A^{[-i_0-m_1]})_{{\rm Im}_1}$ is uniquely determined from~\eqref{Apr3}.

Finally, we will determine $(D^{[-i_0]})_{\Ker_2}$ and $(D^{[-i_0-m_1]})_{S}$. If $m_1\nmid i_0$, for the degree reason, 
we have $(D^{[-i_0]})_{\Ker_2}=(D^{[-i_0-m_1]})_S=0$. Then equation~\eqref{rec4''} with $i=i_0+m_1$ is satisfied and the right-hand side of \eqref{rec4''} with $i=i_0+2m_1$ belongs to ${\rm Im}\,\ad_{G_4}$. 
If $m_1|i_0$,
write
$$
(D^{[-i_0]}(z,s))_{\Ker_2}=\sum_{j=0}^{m_2-1}\beta_j(z,s)s^{j}G_4(z,s)^j
$$
for the coefficients $\beta_j(z,s)\in\CC(s)((z^{-1}))$, $j=0,\dots,m_2-1$, to be determined. 
By assumption, the right-hand side of \eqref{rec4''} with $i=i_0+m_1$ belongs to ${\rm Im}\,\ad_{G_4}$, from which 
 we find that $(D^{[-i_0-m_1]}(z,s))_{S}$ has the form
$$
(D^{[-i_0-m_1]}(z,s))_{S}=\sum_{j=0}^{m_2-1} \beta_j(z,s) U_j(z,s) + U(z,s),
$$
with specific elements $U_j(z,s), U(z,s)\in \D$.

Denote the right-hand side of~\eqref{rec4''} with $i=i_0+2m_1$ as $f_2(-i_0-2m_1,z,s)$. 
We now require that $f_2(-i_0-2m_1,z,s)$ belongs to ${\rm Im}\,\ad_{G_4(z,s)}$. This is equivalent to requiring
\beq\label{orD}
\Tr \,\bigl(f_2(-i_0-2m_1,z,s)G_4(z,s)^\ell\bigr)=0,\quad \forall\,\ell=0,\dots,m_2-1.
\eeq
By a direct computation we find that~\eqref{orD} are equivalent to
\beq
- \sum_{j=0}^{m_2-1}\delta_{\ell+j,m_2-1} \Bigl(\frac{i_0}{m_1}-1+q_a+\frac{\ell}{m_2}\Bigr)m_2 s^{\ell}\beta_j(z,s) + c_\ell(z,s)=0,
\eeq 
where $c_\ell(z,s) \in \CC(s)((z^{-1}))$ had been determined.  
It follows that there exist unique $\beta_0(z,s)$, \dots, $\beta_{m_2-1}(z,s)$ satisfying equations \eqref{orD}. 
Therefore, $(D^{[-i_0]})_{\Ker_2}$ and $(D^{[-i_0-m_1]})_S$ are uniquely determined, 
equation~\eqref{rec4''} with $i=i_0+m_1$ holds, and the right-hand side~\eqref{rec4''} with $i=i_0+2m_1$ belongs to ${\rm Im}\,\ad_{G_4}$.

This concludes the inductive step and thus completes the proof of existence and uniqueness of solutions \eqref{form'} to \eqref{TE'}.
\end{proof}

We now prove Theorem~$\ref{prop1}'$.
\begin{proof}[Proof of Theorem~$\ref{prop1}'$]
For each $a=1,\dots,l-1$, it is easy to check that the series $Y_a(z,s)$ given in Lemma~\ref{thm1''} satisfies~\eqref{form'}. 
This proves the existence part of Theorem~$\ref{prop1}'$. 

For each $a=1,\dots,l-1$, starting from the initial condition~\eqref{form'}, we can check that $Y_a(z,s)$ determined by equation~\eqref{TE'} 
satisfies the initial conditions \eqref{form''1}--\eqref{form''2} in Lemma~\ref{thm1''}.
Then the uniqueness part of Theorem~$\ref{prop1}'$ follows from that of Lemma~\ref{thm1''}. 
\end{proof}

\begin{proof}[Proof of Theorem~$\ref{prop1}$] Follows from Theorem~$\ref{prop1}'$.
\end{proof}

Define two $l\times l$ constant matrices $\eta_1(l)$ and $\eta_2(m_1,m_2)$ by 
\beq
\eta_1(l)=\sum_{i=1}^l e_{i,l+1-i},\qquad \eta_2(m_1,m_2)=\sum_{i=1}^{m_1}e_{i,m_2+i}-\sum_{i=m_1+1}^l e_{i,i-m_2}.
\eeq

\begin{prop}\label{prop-symm1}
If $M(z,s)$ is a solution to the TDE of $(m_1,m_2)$-type~\eqref{TE}, then $\widetilde M(z,s)$ defined by 
$$
\widetilde M(z,s):=\eta_1(l)^{-1}M(-z,-s)\eta_1(l)
$$
is a solution to the TDE of $(m_2,m_1)$-type.
\end{prop}
\begin{proof}
Since $M(z,s)$ satisfies the TDE of $(m_1,m_2)$-type~\eqref{TE}
we have
$$
M(z,s)W(z,s;m_1,m_2)^{-1}=W(z,s;m_1,m_2)^{-1}M(z-1,s).
$$
Here we use the notation $W(z,s;m_1,m_2):=W(z,s)$ to emphasize its dependence in $m_1,m_2$.
It then follows from the definition of $\widetilde M(z,s)$ that 
\beq\label{symm1}
\widetilde M(z-1,s)s^2\eta_1(l)^{-1}W(1-z,-s;m_1,m_2)^{-1}\eta_1(l)=s^2\eta_1(l)^{-1}W(1-z,-s;m_1,m_2)^{-1}\eta_1(l)\widetilde M(z,s).
\eeq
Noticing that 
$$
s^2\eta_1(l)^{-1}W(1-z,-s;m_1,m_2)^{-1}\eta_1(l)=W(z,s;m_2,m_1),
$$
we then get 
$$
\widetilde M(z-1,s)W(z,s;m_2,m_1)=W(z,s;m_2,m_1)\widetilde M(z,s).
$$
The proposition is proved.
\end{proof}

Let $M_a(z,s;m_1,m_2)$, $a=1,\dots,l-1$, denote the solutions to the TDE of $(m_1,m_2)$-type obtained in Theorem~\ref{prop1}. 
For the pair of positive integers $(m_1,m_2)$ we will also use the notations $q_{a;m_1,m_2}=q_a$ and  
$K_{a;m_1,m_2}=K_a$ to emphasize the dependence in $m_1,m_2$.  
We have the following corollary.
\begin{cor}
For each $a=1,\dots, l-1$, the following identity holds:
\begin{align}\label{sol-symm1}
M_a(z,s;m_2,m_1)=(-1)^{q_{a;m_2,m_1}}\eta_1(l)^{-1} M_{l-a}(-z,-s;m_1,m_2)\eta_1(l)+I_{l}\delta_{a,m_2}.
\end{align}
\end{cor}
\begin{proof}
Proposition \ref{prop-symm1} implies the right-hand side of \eqref{sol-symm1} satisfies the TDE of $(m_2,m_1)$-type. 
Since 
$$
M_a(z,s;m_1,m_2)=z^{1-q_{a;m_1,m_2}}(K_{a;m_1,m_2}+O(z^{-1})),
$$
we know that the right-hand side of~\eqref{sol-symm1} has the form
$$
z^{1-q_{l-a;m_1,m_2}}\bigl((-1)^{q_{a;m_2,m_1}-q_{l-a;m_1,m_2}+1}\eta_1(l)^{-1}K_{l-a;m_1,m_2}\eta_1(l)+I_{l}\delta_{a,m_2}+O(z^{-1})\bigr), 
$$ 
which simplifies to 
$$
z^{1-q_{a;m_2,m_1}}(K_{a;m_2,m_1}+O(z^{-1}))
$$
due to the symmetries 
$$
q_{a;m_2,m_1}=q_{l-a;m_1,m_2},\quad K_{a;m_2,m_1}=-\eta_1(l)^{-1}K_{l-a;m_1,m_2}\eta_1(l)+I_{l}\delta_{a,m_2}.
$$
The corollary is then proved by using 
the uniqueness given in Theorem~\ref{prop1}. %, we obtain~\eqref{sol-symm1}.
\end{proof}

\begin{prop}\label{prop-symm2}
If $M(z,s)$ is a solution to the TDE of $(m_1,m_2)$-type, then
$$
\widetilde M(z,s):=\eta_2(m_1,m_2)^{-1}M(z,s)^T\eta_2(m_1,m_2)
$$
is a solution to the TDE of $(m_2,m_1)$-type.
\end{prop}
\begin{proof}
Since $M(z,s)$ satisfies~\eqref{TE}, 
we have
$$
M(z-1,s)^T{W(z,s;m_1,m_2)^{-1}}^T={W(z,s;m_1,m_2)^{-1}}^TM(z,s)^T.
$$
It follows that 
\begin{align}
&\widetilde M(z-1,s)s^2\eta^{-1}_2(m_1,m_2){W(z,s;m_1,m_2)^{-1}}^T\eta_2(m_1,m_2)\nn\\
&=s^2\eta_2(m_1,m_2)^{-1}{W(z,s;m_1,m_2)^{-1}}^T\eta_2(m_1,m_2)\widetilde M(z,s). \nn
\end{align}
The proposition is proved by noticing  
$$
s^2\eta_2(m_1,m_2)^{-1}{W(z,s;m_1,m_2)^{-1}}^T\eta_2(m_1,m_2)=W(z,s;m_2,m_1).
$$
\end{proof}
\begin{cor}\label{cor-symm2}
For each $a=1,\dots,l-1$, the following identity holds:
\begin{align}\label{sol-symm2}
M_{a}(z,s;m_2,m_1) = -\eta_2(m_1,m_2)^{-1}M_{l-a}(z,s;m_1,m_2)^T\eta_2(m_1,m_2)+I_{l}\delta_{a,m_2}.
\end{align}
\end{cor}
\begin{proof}
Proposition \ref{prop-symm2} implies the right-hand side of \eqref{sol-symm2} is a solution to the TDE of $(m_2,m_1)$-type. 
Since 
$$
M_a(z,s;m_1,m_2)=z^{1-q_{a;m_1,m_2}}(K_{a;m_1,m_2}+O(z^{-1})),
$$
we know that the right-hand side of~\eqref{sol-symm2} has the form
$$
z^{1-q_{l-a;m_1,m_2}}\bigl(-\eta_2(m_1,m_2)K_{l-a; m_1,m_2}^T\eta_2(m_1,m_2)+I_{l}\delta_{a,m_2}+O(z^{-1})\bigr), 
$$ 
which simplifies to 
$$
z^{1-q_{a;m_2,m_1}}(K_{a;m_2,m_1}+O(z^{-1})).
$$
%Since the notations have the symmetries
by employing
$$
q_{a;m_2,m_1}=q_{l-a;m_1,m_2},\quad K_{a;m_2,m_1}=-\eta_2(m_1,m_2)K_{l-a;m_1,m_2}^T\eta_2(m_1,m_2)+I_{l}\delta_{a,m_2}.
$$
%the right hand side of \eqref{sol-symm2} has the form
The corollary is then proved by using 
the uniqueness given in Theorem~\ref{prop1}. 
\end{proof}

\begin{remark}
Since
$$
\mathbb{P}^1_{m_1,m_2}\cong\mathbb{P}^1_{m_2,m_1}
$$
 and 
$$
\langle\tau_{i_1}(\phi_{a_1})\dots\tau_{i_k}(\phi_{a_k})\rangle_{g,d}^{\mathbb{P}^1_{m_1,m_2}}
=
\langle\tau_{i_1}(\phi_{l-a_1})\dots\tau_{i_k}(\phi_{l-a_k})\rangle_{g,d}^{\mathbb{P}^1_{m_2,m_1}},
$$
we know that for $k\ge1$ and $a_1,\dots,a_k=1,\dots,l-1$,
$$
F_{a_1,\dots,a_k}(\lambda_1,\dots,\lambda_k;Q;\e;m_1,m_2)=F_{l-a_1.\dots,l-a_k}(\lambda_1,\dots,\lambda_k;Q;\e;m_2,m_1).
$$
By using Corollary~\ref{cor-symm2} one can easily check that the right-hand sides of the conjectural formulas~\eqref{npoint}, \eqref{1point}
do have the corresponding symmetries under the switch of $m_1,m_2$. 
\end{remark}

Before ending this section, recall that 
the dual topological ODE for~$\mathbb{P}^1$ was introduced in~\cite{DYZ}, for which we now give a generalization.
The dual topological ODE for $\mathbb{P}^1_{m_1,m_2}$ for a matrix-valued function $\widehat M=\widehat M(y,s)$ is defined by
\beq\label{dual}
e^y\widehat M W_0-W_0\widehat M =e^y\biggl(\widehat M+\frac{d \widehat M}{d y}\biggr)W_1-W_1\frac{d\widehat M}{d y}
\eeq
where
\beq
W_0=\frac12 e_{1,m_1} + s e_{1,l} - s\sum_{i=2}^le_{i,i-1},\quad
W_1=e_{1,m_1}.
\eeq
Topological and dual topological equations \eqref{TE}, \eqref{dual} are related by a Laplace-type transform:
\beq
\widehat M(y,s)=\frac{1}{2\pi i}\int_\gamma e^{zy}M(z,s)d z
\eeq
where $\gamma$ is an appropriate contour on the complex $z$ plane.

\section{Explicit formulas}\label{explict}
In this section, we give explicit formulas for the unique solutions to the TDE given in Theorem~\ref{prop1}.

A solution $\psi=\psi(z,s)$ to the following linear difference equation
\beq\label{wave-eq}
\psi(z-m_1,s)-\frac1{s}\Bigl(z-\frac12\Bigr)\psi(z,s)+\psi(z+m_2,s)=0
\eeq
is called a {\it quasi-wave function}. 
Similar to~\cite{DYZ, DYZ3} (cf.~\cite{DYZ2, Yang}), by solving~\eqref{wave-eq} we can 
obtain two explicit formal solutions given by the following proposition.
\begin{prop}\label{proppsiAB}
The $\psi_A=\psi_A(z,s;m_1,m_2)$ and $\psi_B=\psi_B(z,s;m_1,m_2)$ given by 
\begin{align}
\psi_A&:=\bigl(\frac s{m_1}\bigr)^{\frac {z-\frac12}{m_1}}\sum_{j\ge0}\frac{(-1)^j s^{\frac{l}{m_1}j}}{{m_1}^{\frac {m_2}{m_1}j}m_2^j j!}\frac1{\Gamma\bigl(\frac{z-\frac12+m_2 j}{m_1}+1\bigr)},\label{psiA}\\
\psi_B&:=\bigl(\frac{s}{m_2}\bigr)^{-\frac{z-\frac12}{m_2}}\sum_{j\ge0}\frac{s^{\frac{l}{m_2}j}}{m_1^j m_2^{\frac{m_1}{m_2}j}j!}\Gamma\Bigl(\frac{z-\frac12-m_1j}{m_2}\Bigr)\label{psiB}
\end{align}
are formal quasi-wave functions. 
Here, the right-hand sides of~\eqref{psiA}, \eqref{psiB} are understood as their asymptotic expansions as $z\to+\infty$.\footnote{
The right-hand sides of~\eqref{psiA}, \eqref{psiB} also have analytic meanings, which will be studied elsewhere.}
\end{prop}
\begin{proof}
By a straightforward verification.
\end{proof}
Note that by using the Stirling formula we can find the space where the formal functions $\psi_A$ and~$\psi_B$ belong to as stated below:
\begin{align}
&\psi_A=\sqrt{\frac{m_1}{2\pi z}}e^{\frac{z}{m_1}}\bigl(\frac{s}{z}\bigr)^{\frac{z-1/2}{m_1}}(1+O(z^{-\frac1{m_1}}))
\in z^{-\frac12}(\frac{es}{z})^{\frac {z-1/2}{m_1}} \cdot \CC\bigl(s^{\frac1{m_1}}\bigr)\bigl(\bigl(z^{-\frac1{m_1}}\bigr)\bigr), \label{mod1}\\
&\psi_B=\sqrt{\frac{2\pi m_2}{z}}e^{-\frac{z}{m_2}}\bigl(\frac{es}{z}\bigr)^{\frac{1/2-z}{m_2}}(1+O(z^{-\frac{1}{m_2}}))
\in z^{-\frac12}(\frac{es}{z})^{\frac {1/2-z}{m_2}}\cdot \CC\bigl(s^{\frac1{m_2}}\bigr)\bigl(\bigl(z^{-\frac1{m_2}}\bigr)\bigr). \label{mod2}
\end{align}

\begin{remark}
For the case when $m_1=m_2=1$, the above formulas~\eqref{psiA}, \eqref{psiB} specialize to
\begin{align}
\psi_A(z,s;1,1)&=s^{z-\frac12}\sum_{j\ge0}\frac{(-1)^l s^{2j}}{ j! \, \Gamma(z+\frac12+ j)}=J_{z-\frac12}(2s), \\
\psi_B(z,s;1,1)&=s^{\frac12-z}\sum_{j\ge0}\frac{s^{2j}}{j!}\Gamma\bigl(z-\frac12-j\bigr)=J_{\frac12-z}(2s) \,
\Gamma\bigl(\frac32-z\bigr)\Gamma\bigl(z-\frac12\bigr),
\end{align}
which agree with~\cite[Proposition 3]{DYZ} and \cite{DYZ3, DMNPS, Yang}.
Here $J_{\nu}(y)$ denotes the Bessel function. 
We also note that 
when one of $m_1,m_2$ equals~1, formulas~\eqref{psiA}, \eqref{psiB} were obtained in~\cite{CG, GJY} (see also~\cite{Alexandrov}). 
Finally we note that in the terminology of~\cite{DMNPS} (cf.~\cite{Alexandrov}) equation~\eqref{wave-eq} could be viewed as a {\it quantum spectral curve}, 
and we hope that equations \eqref{psiA}, \eqref{psiB}, \eqref{mod1}, \eqref{mod2} can be helpful for the study from the point of view of 
Chekhov--Eynard--Orantin topological recursion. 
\end{remark}

\begin{remark}
The formal functions $\psi_A$ and $\psi_B$ are proportional 
to the full asymptotic expansions of the following integrals, respectively,  
\begin{align}
& s^{\frac{z-\frac12}{m_1}}\int_{\gamma_A} t^{-z-\frac12} e^{\frac{t^{m_1}}{m_1}-\frac{t^{-m_2}}{m_2}s^{\frac{m_2}{m_1}+1}} dt, \label{intpsiA} \\
& s^{-\frac{z-\frac12}{m_2}}\int_{\gamma_B} t^{-z-\frac12} e^{\frac{t^{m_1}}{m_1}s^{\frac{m_1}{m_2}+1}-\frac{t^{-m_2}}{m_2}} dt, \label{intpsiB}
\end{align}
with $\gamma_A$, $\gamma_B$ being suitable paths on the complex $z$-plane and within suitable sectors as $z\to\infty$. 
For the case when $m_1=1$, formula~\eqref{intpsiA} was obtained in~\cite{Alexandrov}.
In view of~\cite{ADKMV, Alexandrov, BR}, we hope that  
formulas~\eqref{intpsiA}, \eqref{intpsiB}, \eqref{npoint}, \eqref{1point} could be helpful for obtaining Kontsevich-type matrix models 
for GW invariants of $\mathbb{P}^1_{m_1,m_2}$ without insertions of decendents of $\phi_0=1\in H^0(\mathbb{P}^1_{m_1,m_2})$; 
for the case when $m_1=m_2=1$ this was done in~\cite{Alexandrov, BR}, 
and for the case when one of $m_1, m_2$ equals~1 this in the 
$[{\rm pt}]$-sector should already follow from \cite[Theorem~2]{Alexandrov} and~\cite{CG}.
\end{remark}

Introduce 
$$
\psi_j(z,s;m_1,m_2):=\left\{
\begin{array}{ll}
\psi_A\bigl(z, e^{2\pi \sqrt{-1}(j-1)}s;m_1,m_2\bigr), &j=1,\dots,m_1,\\
\\
\psi_B\bigl(z, e^{2\pi \sqrt{-1}(l-j)}s;m_1,m_2\bigr), &j=m_1+1,\dots,l,
\end{array}\right.
$$
and define a matrix $\Psi(z,s)=(\Psi_{ij}(z,s))_{i,j=1,\dots,l}$ by
\beq\label{defPsi}
\Psi_{ij}(z,s)=\psi_j(z-m_1+i,s;m_1,m_2).
\eeq
Then by a direct calculation we obtain the following lemma.
\begin{lemma}\label{Psiequation} There holds that 
\beq\label{psi}
\Psi(z-1,s)=\frac{1}{s}W(z,s)\Psi(z,s),
\eeq
where $W(z,s)$ is the matrix defined in~\eqref{Q}.
\end{lemma}

Define a matrix $\Phi(z,s)=(\Phi_{ij}(z,s))_{i,j=1,\dots,l}$ by
\beq\label{psi-inv}
\Phi_{ij}(z,s)=\left\{\begin{array}{ll}
\phi_i(z+j),& j=1,\dots,m_1,\\
-\phi_i(z+j-l),& j=m_1+1,\dots,l,
\end{array}\right.
\eeq 
with
$$
\phi_i(z)=\Bigl(\frac{s}{m_1^2}\delta_{i\le m_1}-\frac{s}{m_2^2}\delta_{i>m_1}\Bigr) \, \psi_{l+1-i}(z,s;m_2,m_1),\quad i=1,\dots,l.
$$
Then it is easy to check that
$$
\Psi(z,s)\Phi(z,s)=I_l.
$$
Namely, we have
\begin{lemma}
$$
\Psi(z,s)^{-1}=\Phi(z,s).
$$
\end{lemma}

Introduce
\beq
P_a(s):=\left\{\begin{array}{ll}
s^{1-\frac{a}{m_1}}\diag\bigl(1,\xi_{m_1}^{-a},\dots,\xi_{m_1}^{-(m_1-1)a},\overbrace{0,\dots,0}^{m_2}\bigr),&a=1,\dots, m_1,\\
\\
-s^{\frac{a-m_1}{m_2}}\diag\bigl(\overbrace{0,\dots,0}^{m_1},\xi_{m_2}^{(m_2-1)(a-l)},\dots,\xi_{m_2}^{a-l},1\bigr),&a=m_1+1,\dots, l-1,
\end{array}\right.
\eeq
where 
\beq \label{rootofunity}
\xi_{m_i}:=e^{\frac{2\pi \sqrt{-1}}{m_i}}, \quad i=1,2.
\eeq
Similar to~\cite{BDY1, DYZ2} let us prove the following proposition.
\begin{prop}\label{thm-form}
The unique formal solutions $M_a(z,s)$ to the TDE~\eqref{TE} given in Theorem~\ref{prop1} 
satisfy
\beq\label{exp}
M_a(z,s)=\Psi(z,s)P_a(s){\Psi}(z,s)^{-1}, \quad a=1,\dots,l-1.
\eeq
\end{prop}
\begin{proof}
Denote $\widetilde M_a(z,s):=\Psi(z,s)P_a(s){\Psi}(z,s)^{-1}$. 
Using Lemma~\ref{Psiequation} it is easy to show that $\widetilde M_a(z,s)$ satisfies the TDE. 
%\beq
%\widehat M_a(z-1;s)W(z;s)=\Psi(z-1;s)P_a(s)\Psi(z-1;s)^{-1}(z;s)W(z;s)=W(z;s)\widehat M_a(z;s).
%\eeq
Using the definitions~\eqref{defPsi} and \eqref{psi-inv} we find that
\beq\label{expression-hat}
\widetilde M_{a}(z,s)_{ij}={\left\{\begin{array}{ll}
\widetilde g_a(z,i,j),&j\le m_1,\\
-\widetilde g_a(z,i,j-m_1-m_2),&j> m_1,\\
\end{array}\right.}
\eeq 
where  
\beq\label{expression-gamma1}
\widetilde g_a(z,i,j)=\sum_{k_1\ge0}s^{1+\frac{lk_1+i-j-a}{m_1}}\frac{\sum_{k_3=0}^{m_1-1}\xi_{m_1}^{(m_2k_1+i-j-a)k_3}}{m_1^{\frac{i-j+m_2k_1}{m_1}+1}m_2^{k_1}}\sum_{k_2=0}^{k_1}\frac{(-1)^{k_2}}{k_2!(k_1-k_2)!}\frac{\Gamma(\frac{z-\frac12+j-m_2(k_1-k_2)}{m_1})}{\Gamma(\frac{z-\frac12+i+m_2k_2}{m_1})}
\eeq
for $a=1,\dots, m_1$, and 
\beq\label{expression-gamma2}
\widetilde g_a(z,i,j)=\sum_{k_1\ge0}s^{1+\frac{lk_1-i+j+a}{m_2}} \frac{\sum_{k_3=0}^{m_2-1}\xi_{m_2}^{(m_1k_1-i+j+a)k_3}}{m_1^{k_1}m_2^{2+\frac{m_1(k_1+1)-i+j}{m_2}}}\sum_{k_2=0}^{k_1}\frac{(-1)^{k_2}}{k_2!(k_1-k_2)!}\frac{\Gamma(\frac{z-\frac12+i-m_1(k_1-k_2+1)}{m_2})}{\Gamma(\frac{z-\frac12+j+m_1k_2}{m_2}+1)}
\eeq
for $a=m_1+1,\dots,l-1$. 
It follows from the Stirling formula 
that $\widetilde M_a(z,s)=z^{1-q_a}(H_a+O(z^{-1}))\in z^{1-q_a}\cdot\LLL$.
Thus by Theorem~\ref{prop1} we have $M_a(z,s)=\widetilde M_a(z,s)$.
\end{proof}

\begin{proof}[Proof of Theorem~\ref{thmexpM}]
The theorem follows from Proposition~\ref{thm-form} and the 
well-known formula:
\beq
\frac{\Gamma(z+a)}{\Gamma(z+b)}\sim z^{a-b}\sum_{\ell\ge0}\binom{a-b}{\ell}\frac{B_\ell(a-b+1,a)}{z^\ell} \qquad \mbox{as } z\to +\infty.
\eeq
\end{proof}

The following two corollaries are straightforward from Proposition~\ref{thm-form}.
\begin{cor}\label{property}
The unique formal solutions $M_a(z,s)$ given in Theorem~\ref{prop1} have the following properties: 
\beq
\Tr \, M_a(z,s) = m_1\delta_{a,{m_1}},\quad\det M_a(z,s)=0,\qquad a=1,\dots,l-1.
\eeq
Moreover, 
\beq
M_a(z,s)=\left\{
\begin{array}{ll}
s^{1-a}M_1(z,s)^{a}, & \mbox{ for } a=1,\dots,m_1,\label{Mpower}\\
(-s)^{1-l+a}M_{l-1}(z,s)^{l-a}, & \mbox{ for } a=m_1+1,\dots,l-1,
\end{array}
\right.
\eeq
and
\beq
M_a(z,s)M_b(z,s)=M_b(z,s)M_a(z,s)=0\quad \mbox{ for } a\le m_1<b.
\eeq
\end{cor}

\begin{cor}
We have
$M_a(z, s)\in z^{1-q_a}{\rm Mat}(l\times l,\mathbb{Q}[s][[z^{-1}]])$, $a=1,\dots,l-1$.
\end{cor}

We note that 
when $m_1=m_2$ the expressions for $g_a(z,i,j)$ can be further simplified as follows: 

(i) $g_a(z,i,j)$ vanish unless $m_1|(i-j-a)$; 

(ii) when $m_1|(i-j-a)$, write $p=(i-j-a)/m_1$, then
\begin{align}
&g_a(z,i,j)=z^{-\frac a{m_1}}\frac{s^{p+1}}{m_1^p}
\sum_{k_1\ge-1}
\frac{m_1^{k_1}}{z^{k_1}}
\sum_{k_2=0}^{\lfloor\frac{k_1-p}{2}\rfloor}
\frac{s^{2k_2}}{m_1^{2k_2}}\nn\\
&\times\binom{\frac{j-i}{m_1}-2k_2}{k_1-p-2k_2}
\binom{\frac{i-j}{m_1}+2k_2-1}{k_2}
 B_{k_1-p-2k_2}\Bigl(\frac{j-i}{m_1}-2k_2+1,\frac{j-\frac12}{m_1}-k_2\Bigr)\label{1eq2-ma}
\end{align}
for $a=1,\dots,m_1$, and
\begin{align}
&g_a(z,i,j)=z^{\frac a{m_1}-2}\frac{m_1^{p}}{s^{p-1}}
\sum_{k_1\ge-1}
\frac{m_1^{k_1}}{z^{k_1}}
\sum_{k_2=0}^{\lfloor\frac{k_1+p}{2}\rfloor}
\frac{s^{2k_2}}{m_1^{2k_2}}\nn\\
&\times\binom{\frac{i-j}{m_1}-2k_2-2}{k_1+p-2k_2}
\binom{\frac{j-i}{m_1}+2k_2+1}{k_2}
B_{k_1+p-2k_2}\Bigl(\frac{i-j}{m_1}-2k_2-1,\frac{i-\frac12}{m_1}-k_2-1\Bigr)\label{1eq2-ma2}
\end{align}
for $a=m_1+1,\dots,2m_1-1$.

Using \eqref{expression}, \eqref{expression1}, \eqref{expression2}, \eqref{expression-gamma1}, \eqref{expression-gamma2} and \eqref{1point}, we obtain explicit 1-point functions given in the following two propositions.
\begin{prop}[*]\label{gf-1point}
For $a=1,\dots,m_1$, we have
\begin{align}
&F_{a}(\lambda;Q;\e)=-\e^{-1}\delta_{a,m_1}\bigl(\psi(\frac{\lambda}{\e}+\frac12)-\log(\frac\lambda\e)\bigr)\nn\\
&-\sum_{k_1\ge0}
\delta_{m_1|(m_2k_1-a)}
\frac{(-1)^{k_1}Q^{m_2k_1+(1-q_a)\rho}\e^{q_a-2-\frac{lk_1}{m_1}}}{m_1^{\frac{m_2k_1}{m_1}+1}m_2^{k_1}k_1!}
\sum_{k_2\ge0}
(-1)^{\lfloor\frac{k_2}{m_2}\rfloor}\binom{k_1-1}{\lfloor\frac{k_2}{m_2}\rfloor}
\frac{\Gamma(\frac{\frac\lambda\e-\frac12-k_2}{m_1})}{\Gamma(\frac{\frac\lambda\e-\frac12-k_2}{m_1}+\frac{m_2k_1}{m_1}+1)},\label{1point1}
\end{align}
 and for $a=m_1+1,\dots,l-1$,
\begin{align}
F_{a}(\lambda;Q;\e)=&-\sum_{k_1\ge0}
\delta_{m_2|(m_1k_1-l+a)}
\frac{(-1)^{k_1}Q^{m_1k_1+(1-q_a)\rho}\e^{q_a-2-\frac{lk_1}{m_2}}}{m_2^{\frac{m_1k_1}{m_2}+1}m_1^{k_1}k_1!}\nn\\
&\times\sum_{k_2\ge0}
(-1)^{\lfloor\frac{k_2}{m_1}\rfloor}\binom{k_1-1}{\lfloor\frac{k_2}{m_1}\rfloor}
\frac{\Gamma(\frac{\frac\lambda\e-\frac12-k_2}{m_2})}{\Gamma(\frac{\frac\lambda\e-\frac12-k_2}{m_2}+\frac{m_1k_1}{m_2}+1)}.\label{1point2}
\end{align}
\end{prop}
\begin{prop}[*]
For $a=1,\dots,m_1$, we have
\begin{align}
F_{a}(\lambda;Q;\e)=&\delta_{a,m_1}\sum_{g\ge0}\frac{\e^{2g-1}}{\lambda^{2g}}\frac{1-2^{2g-1}}{2^{2g}g}B_{2g}
-\sum_{k_1\ge0}
\delta_{m_1|(m_2k_1-a)}
\frac{(-1)^{k_1}Q^{m_2k_1+(1-q_a)\rho}\e^{q_a-1-k_1}}{m_2^{k_1}k_1!\lambda^{\frac{m_2k_1}{m_1}+1}}\nn\\
&\times\sum_{k_2\ge0}\frac{m_1^{k_2}\e^{k_2}}{\lambda^{k_2}}\binom{-\frac{m_2k_1}{m_1}-1}{k_2}\sum_{k_3\ge0}
(-1)^{\lfloor\frac{k_3}{m_2}\rfloor}\binom{k_1-1}{\lfloor\frac{k_3}{m_2}\rfloor}
B_{k_2}\Bigl(-\frac{m_2k_1}{m_1},-\frac{k_3+\frac12}{m_1}\Bigr),\label{1point1-2}
\end{align}
and for $a=m_1+1,\dots,l-1$,
\begin{align}
F_{a}(\lambda;Q;\e)=&
-\sum_{k_1\ge0}
\delta_{m_2|(m_1k_1-l+a)}
\frac{(-1)^{k_1}Q^{m_1k_1+(1-q_a)\rho}\e^{q_a-1-k_1}}{m_1^{k_1}k_1!\lambda^{\frac{m_1k_1}{m_2}+1}}
\sum_{k_2\ge0}\frac{m_2^{k_2}\e^{k_2}}{\lambda^{k_2}} \nn\\
&\times\binom{-\frac{m_1k_1}{m_2}-1}{k_2}\sum_{k_3\ge0}
(-1)^{\lfloor\frac{k_3}{m_1}\rfloor}\binom{k_1-1}{\lfloor\frac{k_3}{m_1}\rfloor}
B_{k_2}\Bigl(-\frac{m_1k_1}{m_2},-\frac{k_3+\frac12}{m_2}\Bigr).\label{1point2-2}
\end{align}
\end{prop}

The following corollary is straightforward.
\begin{cor}[*]
The 1-point degree~0 numbers have the expressions:
\beq
\langle \tau_{i}(\phi_{a}) \rangle_{g, \, 0}=
\delta_{a,m_1}\delta_{i,2g-2}\frac{1-2^{2g-1}}{2^{2g-1}(2g)!}B_{2g}.
\eeq
\end{cor}

With the help of the following identity 
\beq
\sum_{j=0}^k(-1)^j\binom{k}{j}B_m(\ell,x-j)=\frac{m!}{(m-k)!}B_{m-k}(\ell-k,x-k),
\eeq
from Proposition~\ref{gf-1point} we can also obtain the following corollary.
%Then when $m_1=m_2$, the $1$-point functions $F_a(\lambda;Q;e)$ can be simplified as follow.}
\begin{cor}[*]\label{gf-1point-1eq2}
When $m_1=m_2$,  the 1-point numbers $\langle\tau_i(\phi_{a})\rangle_{g, \,d}$ vanish for $a=1,\dots,m_1-1$ and for $a=m_1+1,\dots,l-1$. Moreover,
\begin{align}
&F_{m_1}(\lambda;Q;\e)=\sum_{g\ge0}\frac{\e^{2g-1}}{\lambda^{2g}}\frac{1-2^{2g-1}}{2^{2g}g}B_{2g}
-\sum_{k_1\ge0}
\frac{(-1)^{k_1}Q^{m_1k_1}\e^{-k_1}}{m_1^{k_1}{k_1!}^2\lambda^{k_1+1}}\nn\\
&\times\sum_{k_2\ge0}\frac{(-1)^{k_2}m_1^{k_2}\e^{k_2}}{\lambda^{k_2}}\frac{(k_1+k_2)!}{(k_2-k_1+1)!}
\sum_{k_3=0}^{m_1-1}
B_{k_2-k_1+1}\Bigl(1-2k_1,1-k_1-\frac{k_3+\frac12}{m_1}\Bigr).
\end{align}
\end{cor}

For the case when $m_1=m_2=1$, one can check that Corollary~\ref{gf-1point-1eq2} agrees with~\cite[(36)]{DYZ}.

\section{Computation of $\langle\tau_{i}(\phi_a)^k\rangle_{g,d}$}\label{polygon}
In this section we do concrete computations for some of the Gromov--Witten invariants of $\mathbb{P}^1_{m_1,m_2}$ 
with $(m_1,m_2)$ being $(2,1)$, $(3,1)$ and $(2,2)$, based on the explicit (conjectural) formulas~\eqref{npoint}, \eqref{1point}.

It will be convenient to use an algorithm described in \cite{DY1}, \cite{DY2}.
Fix $\mathbf{b}=((a_1,i_1),(a_2,i_2),\dots)$ an arbitrary sequence of pairs of non-negative integers with $a_j\in\{1,\dots,l-1\}$, $i_j\in\mathbb{Z}_{\ge0}$. Following~\cite{DY1}, \cite{DY2}, define recursively a family 
of Laurent series $R^\mathbf{b}_{a,K}\in{\rm Mat}(l\times l,\mathbb{Q}[\e]((\lambda^{-1})))$ with $K=\{k_1,\dots,k_m\}$ by
\begin{align}
R^\mathbf{b}_{a,\{\}}(\lambda;\e)&:=\e^{1-q_{a}}\lambda^{q_{a}}M_{a}\Bigl(\frac{\lambda}{\e}, \frac1\e\Bigr),\\
R^\mathbf{b}_{a,K}(\lambda;\e)&:=\sum_{I\sqcup J=K\setminus\{k_1\}}\Bigl[\bigl(\lambda^{i_{k_1}}R^\mathbf{b}_{a_{k_1},I}\bigr)_+,R^\mathbf{b}_{a,J}\Bigr].
\end{align}
Here $k_1,\dots,k_m$ are distinct positive integers, and $M_a(z,s)$ are the unique solutions to the TDE~\eqref{TE} satisfying~\eqref{form}.
For the case when $(a_1,i_1)=(a_2,i_2)=\dots=(a,i)$, like in~\cite{DY1, DY2}, we have
\beq
R^\mathbf{b}_{b,K}(\lambda)=R^\mathbf{b}_{b,K'}(\lambda)=:R^{(a,i)}_{b,m}(\lambda),\quad\text{as long as}\ |K|=|K'|,
\eeq
and
\beq
R^{(a,i)}_{b, m}=\sum_{\ell=0}^{m-1}\binom{m-1}\ell\Bigl[\bigl(\lambda^i R^{(a,i)}_{a,\ell}\bigr)_+,R^{(a,i)}_{b, m-1-\ell}\Bigr], \quad m\ge1.
\eeq

The following proposition follows using the arguments given in \cite{DY1, DY2}.

\begin{prop}[$^*$]
Let $\mathbf{b}=((a_1,i_1),(a_2,i_2),\dots)$ and $K=\{k_1,\dots, k_m\}$. The following formula holds true:
\begin{align}
&\sum_{j_1,j_2\ge0}\frac{q_{b,j_1}q_{c,j_2}\prod_{\ell=1}^{m} {q_{a,i_\ell}}}{\lambda^{j_1+1}\mu^{j_2+1}}\sum_{g\ge0}\sum_{d\ge0}\e^{2g+m}
\langle\tau_{j_1}(\phi_{b})\tau_{j_2}(\phi_c)\prod_{\ell=1}^{m}\tau_{i_{k_\ell}}(\phi_{a_{k_\ell}})\rangle_{g,d}\nn\\
&=\sum_{I\sqcup J=K}\frac{\Tr \, R^{\mathbf{b}}_{b,I}(\lambda;\e)R^{\mathbf{b}}_{c,J}(\mu;\e)}{(\lambda-\mu)^2}
\notag\\
&-\delta_{m,0}
\frac{\delta_{b+c, m_1}
(b\lambda+c\mu)+\delta_{b,m_1} \delta_{c,m_1}m_1\lambda\mu 
+\delta_{b+c, 2m_1+m_2}((l-b)\lambda+(l-c)\mu)}
{ (\lambda-\mu)^2}.\label{100}
\end{align}
Here $m=|K|$. In particular case that $(a_1,i_1)=(a_2,i_2)=\dots=(a,i)$, we have for $m\ge0$
\begin{align}
&\sum_{j_1,j_2\ge0}\frac{q_{a,i}^{m} q_{b,j_1}q_{c,j_2}}{\lambda^{j_1+1}\mu^{j_2+1}}\sum_{g\ge0}\sum_{d\ge0}\e^{2g+m}\langle\tau_{i}(\phi_{a})^m\tau_{j_1}(\phi_{b})\tau_{j_2}(\phi_c)\rangle_{g,d}
=\sum_{\ell=0}^m \binom{m}{\ell}\frac{\Tr \, R^{(a,i)}_{b,\ell}(\lambda)R^{(a,i)}_{c, m-\ell}(\mu)}{(\lambda-\mu)^2}
\notag\\
&-\delta_{m,0}
\frac{\delta_{b+c, m_1}
(b\lambda+c\mu)+\delta_{b,m_1} \delta_{c,m_1}m_1\lambda\mu 
+\delta_{b+c, 2m_1+m_2}((l-b)\lambda+(l-c)\mu)}
{(\lambda-\mu)^2 }. \label{algpolygon}
\end{align}
\end{prop}

Using~\eqref{algpolygon} we now do concrete computations for GW invariants of $\mathbb{P}^1_{m_1,m_2}$ 
of the form 
\beq\label{polygonnumbers}
\langle\tau_{i}(\phi_a)^k\rangle_{g,d}, \quad k\ge2.
\eeq
Here, $i\ge0$ and $a=1,\dots,l-1$. 
The degree-dimension counting now reads
\beq\label{poly-degree-dimension}
2g-2+\frac{d}{\rho}=(i+q_a-1)k.
\eeq

When $m_1=m_2=1$, concrete computations for~\eqref{polygonnumbers} were carried out 
in~\cite{DY2}.

Consider the case $m_1=2, m_2=1$. For $i=0$, $\langle\tau_{0}(\phi_a)^k\rangle_{g,d}$ are primary GW invariants of $\mathbb{P}^1_{2,1}$. We obtain from~\eqref{1point}, \eqref{algpolygon} that
$$
\langle\tau_{0}(\phi_a)^k\rangle_{g,d}=\delta_{a,1}\delta_{k,1}\delta_{g,0}\delta_{d,2}-\frac14\delta_{a,1}\delta_{k,4}\delta_{g,0}\delta_{d,0}-\frac1{24}\delta_{a,2}\delta_{k,1}\delta_{g,1}\delta_{d,0}.
$$
We list in Tables \ref{P21tau11}--\ref{P21tau22} a few more invariants.

Consider the case $m_1=3, m_2=1$.
For $i=0$, we obtain from~\eqref{1point}, \eqref{algpolygon} the following 
$$
\langle\tau_{0}(\phi_a)^k\rangle_{g,d}=
\left\{
{\renewcommand{\arraystretch}{1.5}\begin{array}{ll}
1,&(a,k,g,d)=(1,1,0,1),\\
\frac13,&(a,k,g,d)=(1,3,0,0), (2,2,0,1),\\
-\frac1{27},&(a,k,g,d)=(2,6,0,0),\\
-\frac1{24},&(a,k,g,d)=(3,1,1,0),\\
0,&\text{otherwise}.
\end{array}}
\right.
$$
We list in Tables \ref{P31tau11}--\ref{P31tau32} the first few GW invariants of $\mathbb{P}^1_{3,1}$.

Consider the case $m_1=m_2=2$.
For $i=0$, we obtain from~\eqref{1point}, \eqref{algpolygon} together with a guess work that
$$
\langle\tau_{0}(\phi_a)^k\rangle_{g,d}=
\left\{
{\renewcommand{\arraystretch}{1.5}\begin{array}{ll}
-\frac14,&(a,k,g,d)=(1,3,0,0), (2,2,0,1),\\
-\frac1{24},&(a,k,g,d)=(2,1,1,0),\\
2^{k-1},&(a,k,g,d)=(2,k,0,2),\\
0,&\text{otherwise}.
\end{array}}
\right.
$$
We list in Tables \ref{P22tau11}--\ref{P22tau22} the first few GW invariants of $\mathbb{P}^1_{2,2}$.

In the above examples, we also note that 
 explicit expressions of the potentials of the corresponding Frobenius manifolds 
 can be found in~\cite{DZ} (cf.~\cite{Du2, MT}). 
 One can verify that the computations from~\eqref{string}, \eqref{100} or \eqref{algpolygon} agree with these expressions.

\begin{table}[h!]
\renewcommand{\arraystretch}{1.2}
$$
\begin{array}{|c|c|c|c|c|c|}
\hline
k&g=0&g=1&g=2&g=3&g=4\\
\Xhline{1pt}
1&0&0&0&0&0\\
\hline
2&\frac12&0&0&0&0\\
\hline
3&0&-\frac18&0&0&0\\
\hline
4&0&0&-\frac1{16}&0&0\\
\hline
5&10&0&0&0&0\\
\hline
6&0&0&0&0&0\\
\hline
7&0&0&\frac{735}{64}&0&0\\
\hline
8&1260&0&0&\frac{8625}{128}&0\\
\hline
9&0&0&0&0&0\\
\hline
10&0&0&-\frac{66465}{16}&0&0\\
\hline
11&540540&0&0&-\frac{999075}{8}&0\\
\hline
12&0&259875&0&0&-\frac{4054513925}{2048}\\
\hline
\end{array}
$$
\caption{\label{P21tau11}$\langle\tau_1(\phi_1)^k\rangle_{g,d=(4-4g+k)/3}$ for $\mathbb{P}^1_{2,1}$.}
\end{table}

\begin{table}[h!]
\renewcommand{\arraystretch}{1.2}
$$
\begin{array}{|c|c|c|c|c|c|c|}
\hline
k&g=0&g=1&g=2&g=3&g=4&g=5\\
\Xhline{1pt}
1&\frac12&0&0&0&0&0\\
\hline
2&0&0&0&0&0&0\\
\hline
3&0&\frac12&0&0&0&0\\
\hline
4&12&0&0&0&0&0\\
\hline
5&0&0&\frac12&0&0&0\\
\hline
6&0&480&0&0&0&0\\
\hline
7&6720&0&0&\frac12&0&0\\
\hline
8&0&0&17472&0&0&0\\
\hline
9&0&2016000&0&0&\frac12&0\\
\hline
10&19353600&0&0&629760&0&0\\
\hline
11&0&0&486541440&0&0&\frac12\\
\hline
12&0&23417856000&0&0&22674432&0\\
\hline
\end{array}
$$
\caption{\label{P21tau21}$\langle\tau_1(\phi_2)^k\rangle_{g,d=2(2-2g+k)/3}$ for $\mathbb{P}^1_{2,1}$.}
\end{table}

\begin{table}[h!]
\renewcommand{\arraystretch}{1.2}
$$
\begin{array}{|c|c|c|c|c|}
\hline
k&g=1&g=4&g=7&g=10\\
\Xhline{1pt}
1&\frac1{12}&0&0&0\\
\hline
2&\frac{7}{4}&0&0&0\\
\hline
3&\frac{7}{4}&0&0&0\\
\hline
4&\frac{181}{12}&-\frac{21293}{414720}&0&0\\
\hline
5&\frac{2041}{12}&-\frac{47933}{82944}&0&0\\
\hline
6&2373&-\frac{81187}{13824}&0&0\\
\hline
7&\frac{473797}{12}&-\frac{177821}{9216}&0&0\\
\hline
8&\frac{2289842}{3}&\frac{26295563}{1296}&\frac{115829496601}{7962624}&0\\
\hline
9&\frac{67260123}{4}&\frac{14166735121}{4608}&\frac{5186028997597}{7962624}&0\\
\hline
10&\frac{1247580880}{3}&\frac{5488889021}{16}&\frac{5093893075885}{248832}&0\\
\hline
11&\frac{136912202101}{12}&\frac{453026908622057}{13824}&\frac{169533298949245}{294912}&0\\
\hline
12&343895883552&\frac{103233320612411}{36}&\frac{15632457282359225}{995328}&-\frac{11131036261937986011499}{12230590464}\\
\hline
\end{array}
$$
\caption{\label{P21tau12}$\langle\tau_2(\phi_1)^k\rangle_{g,d=(4-4g+3k)/3}$ for $\mathbb{P}^1_{2,1}$. 
By~\eqref{poly-degree-dimension} these GW invariants with $g\not\equiv1 \,({\rm mod}\, 3)$ vanish.
}
\end{table}

\begin{table}[h!]
\renewcommand{\arraystretch}{1.2}
$$
\begin{array}{|c|c|c|c|c|c|c|}
\hline
k&g=0&g=1&g=2&g=3&g=4&g=5\\
\Xhline{1pt}
1& 0 & 0 & \frac7{5760} & 0 & 0 & 0\\
\hline
2& \frac{1}{4} & 0 & 0 & 0 & 0 & 0\\
\hline
3& 0 & \frac{23}{8} & 0 & 0 & 0 & 0 \\
\hline
4& 0 & 0 & \frac{195}{8} & 0 & 0 & 0 \\
\hline
5& 45 & 0 & 0 & \frac{80795}{432} & 0 & 0 \\
\hline
6& 0 & 6690 & 0 & 0 & \frac{2384437}{1728} & 0\\
\hline
7& 0 & 0 & 670425 & 0 & 0 & \frac{34611451}{3456} \\
\hline
8& 124320 & 0 & 0 & 57254960 & 0 & 0 \\
\hline
9& 0 & 80826480 & 0 & 0 & \frac{13532788570}{3} & 0  \\
\hline
10& 0 & 0 & 34059521160 & 0 & 0 & \frac{1019579947540}{3} \\
\hline
11& 1530144000 & 0 & 0 & 11864055062860 & 0 & 0  \\
\hline
%12& 0 & 2813149785600 & 0 & 0 & 3718644805163750 & 0   \\
%\hline
\end{array}
$$
\caption{\label{P21tau22}$\langle\tau_2(\phi_2)^k\rangle_{g, \, d=4(1-g+k)/3}$ for $\mathbb{P}^1_{2,1}$.
}
\end{table}

\begin{table}[h!]
\renewcommand{\arraystretch}{1.2}
$$
\begin{array}{|c|c|c|c|c|c|}
\hline
k&g=0&g=1&g=2&g=3&g=4\\
\Xhline{1pt}
 2 & \frac{1}{2} & 0 & 0 & 0 & 0 \\
\hline
 4 & 0 & -1 & 0 & 0 & 0 \\
\hline
 6 & 40 & 0 & -\frac{67}{6} & 0 & 0 \\
\hline
 8 & 0 & \frac{2240}{3} & 0 & 0 & 0 \\
\hline
 10 & 16800 & 0 & \frac{490070}{3} & 0 & 0 \\
\hline
 12 & 0 & -6899200 & 0 & 38449565 & 0 \\
\hline
 14 & 134534400 & 0 & -8016848840 & 0 & 0 \\
\hline
 16 & 0 & 264135872000 & 0 & -\frac{272918591545600}{27} & 0 \\
\hline
 18 & -1646701056000 & 0 & 1038311341267200 & 0 & -\frac{270122401234335650}{27} \\
\hline
% 20 & 0 & -29038109388288000 & 0 & \frac{38009060783629936000}{9} & 0 \\
%\hline
\end{array}
$$
\caption{\label{P31tau11}$\langle\tau_1(\phi_1)^k\rangle_{g, \, d=\frac{6-6g+k}4}$ for $\mathbb{P}^1_{3,1}$. By~\eqref{poly-degree-dimension} these GW invariants with odd $k$ vanish.
}
\end{table}

\begin{table}[h!]
\renewcommand{\arraystretch}{1.2}
$$
\begin{array}{|c|c|c|c|c|c|c|}
\hline
k&g=0&g=1&g=2&g=3&g=4&g=5\\
\Xhline{1pt}
 1 & \frac14 & 0 & 0 & 0 & 0 & 0 \\
\hline
 2 & 0 & -\frac{1}{36} & 0 & 0 & 0 & 0 \\
\hline
 3 & \frac{1}{3} & 0 & -\frac{1}{80} & 0 & 0 & 0 \\
\hline
 4 & 0 & \frac{1}{18} & 0 & 0 & 0 & 0 \\
\hline
 5 & \frac{5}{3} & 0 & \frac{251}{1296} & 0 & 0 & 0 \\
\hline
 6 & 0 & \frac{5}{9} & 0 & \frac{34573}{46656} & 0 & 0 \\
\hline
 7 & \frac{182}{9} & 0 & -\frac{3871}{1296} & 0 & 0 & 0 \\
\hline
 8 & 0 & \frac{1610}{81} & 0 & -\frac{43246}{729} & 0 & 0 \\
\hline
 9 & \frac{1400}{3} & 0 & 70 & 0 & -\frac{356307091}{559872} & 0 \\
\hline
 10 & 0 & \frac{23800}{27} & 0 & \frac{22823255}{5832} & 0 & 0 \\
\hline
 11 & \frac{160160}{9} & 0 & -\frac{1744435}{972} & 0 & \frac{125545646303}{839808} & 0 \\
\hline
 12 & 0 & \frac{1641640}{27} & 0 & -\frac{398212045}{1458} & 0 & \frac{1756207031495}{559872} \\
\hline
\end{array}
$$
\caption{\label{P31tau21}$\langle\tau_1(\phi_2)^k\rangle_{g, \,d= \frac{3-3g+k}2}$ for $\mathbb{P}^1_{3,1}$.
}
\end{table}

\begin{table}[h!]
\renewcommand{\arraystretch}{1.2}
$$
\begin{array}{|c|c|c|c|c|c|}
\hline
k&g=0&g=1&g=2&g=3&g=4\\
\Xhline{1pt}
 2 & 1 & 0 & 0 & 0 & 0 \\
\hline
 4 & 0 & 9 & 0 & 0 & 0 \\
\hline
 6 & 1215 & 0 & 81 & 0 & 0 \\
\hline
 8 & 0 & 357210 & 0 & 729 & 0 \\
\hline
 10 & 55112400 & 0 & 86113125 & 0 & 6561 \\
\hline
 12 & 0 & 114578679600 & 0 & 19797948720 & 0 \\
\hline
 14 & 17874274017600 & 0 & 176955312774240 & 0 & 4487187539835 \\
\hline
 %16 & 0 & 145585961873352000 & 0 & 248277496306039200 & 0 \\
%\hline
% 18 & \begin{gathered}221515877900116\\80000\end{gathered} & 0 & \begin{gathered}834524150675706\\691200\end{gathered} & 0 & \begin{gathered}334817557514728\\207920\end{gathered} \\
%\hline
% 20 & 0 & \begin{gathered}515409853745317\\763232000\end{gathered} & 0 & \begin{gathered}420118374906123\\4183248000\end{gathered} & 0 \\
%\hline
\end{array}
$$
\caption{\label{P31tau31}$\langle\tau_1(\phi_3)^k\rangle_{g, \,d=\frac{3(2-2g+k)}4}$ for $\mathbb{P}^1_{3,1}$. By~\eqref{poly-degree-dimension} these GW invariants with odd $k$ vanish.
}
\end{table}

\begin{table}[h!]
\renewcommand{\arraystretch}{1.2}
$$
\begin{array}{|c|c|c|c|c|c|}
\hline
k&g=1&g=3&g=5&g=7\\
\Xhline{1pt}
 1  & \frac{1}{8}  & 0 & 0 & 0   \\
\hline
 2  & \frac{5}{12}  & 0 & 0 & 0   \\
\hline
 3  & \frac{59}{24}  & -\frac{4003}{32256} & 0 & 0   \\
\hline
 4  & 21  & -\frac{15899}{10368} & 0 & 0   \\
\hline
 5  & \frac{5651}{24}  & -\frac{192995}{10368} & 0 & 0   \\
\hline
 6  & 3272  & -\frac{707885}{3456}  & \frac{524958355}{497664} & 0   \\
\hline
 7  & \frac{434225}{8}  & \frac{2163665}{2592}  & \frac{198414344905}{3981312} & 0   \\
\hline
 8  & \frac{3140504}{3}  & \frac{1352411795}{5184}  & \frac{324035145455}{186624} & 0\\
\hline
 9  & \frac{184143297}{8}  & \frac{37338329}{2}  & \frac{2046920979565}{36864}  & -\frac{10906153043084315}{26873856}   \\
\hline
 10  & 568369280  & \frac{1938389986145}{1728}  & \frac{36317187827375}{20736}  & -\frac{36157990087745346245}{859963392}   \\
\hline
 11  & \frac{373745013803}{24}  & \frac{6950713354145}{108}  & \frac{112462281806699825}{1990656}  & -\frac{614392003666296451475}{214990848}   \\
\hline
 12  & 468847405440  & \frac{469915449644355}{128}  & \frac{1127205334606505}{576}  & -\frac{23503283746359067242185}{143327232}  \\
\hline
\end{array}
$$
\caption{\label{P31tau12}$\langle\tau_2(\phi_1)^k\rangle_{g, \,d=\frac{3-3g}2+k}$ for $\mathbb{P}^1_{3,1}$.
By~\eqref{poly-degree-dimension} these GW invariants with even $g$ vanish.}
\end{table}

\begin{landscape}
\begin{table}[h!]
\renewcommand{\arraystretch}{1.2}
$$
\begin{array}{|c|c|c|c|c|c|c|c|}
\hline
k&g=0&g=1&g=2&g=3&g=4&g=5&g=6\\
\Xhline{1pt}
 2 & \frac{1}{16} & 0 & \frac{23}{2880} & 0 & 0 & 0 & 0 \\
\hline
 4 & 0 & \frac{13}{8} & 0 & \frac{211}{1944} & 0 & 0 & 0 \\
\hline
 6 & \frac{55}{8} & 0 & \frac{11635}{72} & 0 & \frac{227260583}{14929920} & 0 & -\frac{1328862557329}{14332723200} \\
\hline
 8 & 0 & \frac{67900}{9} & 0 & \frac{1536532499}{31104} & 0 & \frac{4957207726373}{537477120} & 0 \\
\hline
 10 & 36225 & 0 & \frac{14328728155}{1152} & 0 & \frac{152372243833523}{3981312} & 0 & \frac{209758142480576117}{12899450880} \\
\hline
 12 & 0 & \frac{4868326925}{16} & 0 & \frac{2867708306023715}{82944} & 0 & \frac{1547748390057351251}{23887872} & 0 \\
\hline
 14 & \frac{46495123675}{32} & 0 & \frac{51851459478333515}{18432} & 0 & \frac{1916980112463068253601}{11943936} & 0 & \frac{6628677510549036153630419}{30958682112} \\
\hline
\end{array}
$$
\caption{\label{P31tau22}$\langle\tau_2(\phi_2)^k\rangle_{g,d=(6-6g+5k)/4}$ for $\mathbb{P}^1_{3,1}$. 
By~\eqref{poly-degree-dimension} these GW invariants with odd $k$ vanish.}
\end{table}

\begin{table}[h!]
$$
\renewcommand{\arraystretch}{1.2}
\begin{array}{|c|c|c|c|c|c|c|}
\hline
k&g=0&g=1&g=2&g=3&g=4&g=5\\
\Xhline{1pt}
 1 & \frac16 & 0 & \frac7{5760} & 0 & 0 & 0 \\
\hline
 2 & 0 & \frac{5}{8} & 0 & 0 & 0 & 0 \\
\hline
 3 & \frac{9}{8} & 0 & \frac{237}{128} & 0 & 0 & 0 \\
\hline
 4 & 0 & 54 & 0 & \frac{1305}{256} & 0 & 0 \\
\hline
 5 & \frac{135}{2} & 0 & \frac{121797}{64} & 0 & \frac{111807}{8192} & 0 \\
\hline
 6 & 0 & \frac{112995}{8} & 0 & \frac{957987}{16} & 0 & \frac{1183815}{32768} \\
\hline
 7 & \frac{25515}{2} & 0 & \frac{265054923}{128} & 0 & \frac{3649840803}{2048} & 0 \\
\hline
 8 & 0 & \frac{15079365}{2} & 0 & \frac{33720220863}{128} & 0 & \frac{823455801}{16} \\
\hline
 9 & \frac{10333575}{2} & 0 & \frac{48543139701}{16} & 0 & \frac{127256523683625}{4096} & 0 \\
\hline
 10 & 0 & \frac{54493074405}{8} & 0 & \frac{16618432581135}{16} & 0 & \frac{57304576050330735}{16384} \\
\hline
 11 & 3682886130 & 0 & \frac{767757835806885}{128} & 0 & \frac{20792898489236643}{64} & 0 \\
\hline
 12 & 0 & 9336116339550 & 0 & \frac{1134205470126545355}{256} & 0 & \frac{12276036590917496859}{128} \\
\hline
\end{array}
$$
\caption{\label{P31tau32}$\langle\tau_2(\phi_3)^k\rangle_{g,d=3(1-g+k)/2}$ for $\mathbb{P}^1_{3,1}$.
}
\end{table}
\end{landscape}

%This may comes from that the two obifold points of $\mathbb{P}^1_{2,2}$ both have index $2$.

\begin{landscape}
\begin{table}[h!]
\renewcommand{\arraystretch}{1.2}
$$
\begin{array}{|c|c|c|c|c|}
\hline
k&g=0&g=1&g=2&g=3\\
\Xhline{1pt}
2&0&0&0&0\\
\hline
 4 & 12 & 0 & -\frac{1}{16} & 0 \\
\hline
6&0&0&0&0\\
\hline
 8 & 26880 & 2520 & \frac{385}{4} & \frac{8625}{128}  \\
\hline
10&0&0&0&0\\
\hline
 12 & 558835200 & 116424000 & 821205 & -\frac{24963015}{16}  \\
\hline
14&0&0&0&0\\
\hline
 16 & 50912122060800 & 18852305472000 & 573469696800 & 27060558525 \\
\hline
%18&0&0&0&0\\
%\hline
% 20 & 13908089813409792000 & 8167961637868032000 & 607402743299352000 & -1051940110221000 \\
%\hline
\end{array}
$$
\caption{\label{P22tau11}$\langle\tau_1(\phi_1)^k\rangle_{g,d=2-2g+k/2}=\langle\tau_1(\phi_3)^k\rangle_{g,d=2-2g+k/2}$ for $\mathbb{P}^1_{2,2}$. 
By~\eqref{poly-degree-dimension} these GW invariants with odd $k$ vanish.
}
\end{table}

\begin{table}[h!]
\renewcommand{\arraystretch}{1.2}
$$
\begin{array}{|c|c|c|c|c|c|}
\hline
k&g=0&g=1&g=2&g=3&g=4\\
\Xhline{1pt}
1&0&0&0&0&0\\
\hline
 2 & 1 & \frac{1}{2} & 0 & 0 &0\\
\hline
3&0&0&0&0&0\\
\hline
 4 & 32 & 40 & \frac{1}{2} & 0 &0\\
\hline
5&0&0&0&0&0\\
\hline
 6 & 3840 & 9440 & 1456 & \frac{1}{2} &0\\
\hline
7&0&0&0&0&0\\
\hline
 8 & 1075200 & 4515840 & 2217152 & 52480 &\frac12\\
\hline
9&0&0&0&0&0\\
\hline
 10 & 557383680 & 3645573120 & 3912007680 & 501385280&1889536 \\
\hline
%11&0&0&0&0&0\\
%\hline
% 12 & 467335249920 & 4458078535680 & 8666558208000 & 3143164231680&112764066272 \\
%\hline
%13&0&0&0&0&0\\
%\hline
% 14 & 580400405544960 & 7691173281792000 & 24169510931005440 & 18534502531184640&2475770479902720 \\
%\hline
% 16 & \begin{gathered}100406436172922\\8800\end{gathered} & \begin{gathered}178002909689020\\41600\end{gathered} & \begin{gathered}837298687590767\\00160\end{gathered} & \begin{gathered}115099802961539\\235840\end{gathered}&\begin{gathered}382223091348274\\17600\end{gathered} \\
%\hline
% 18 &\begin{gathered} 231253286584320\\0000000\end{gathered} & \begin{gathered}532296782921531\\91628800\end{gathered} & \begin{gathered}354546527634585\\425018880\end{gathered} & \begin{gathered}785750382265791\\637094400\end{gathered}&\begin{gathered}520620937028565\\026734080\end{gathered} \\
%\hline
% 20 & \begin{gathered}684983795677058\\7633254400\end{gathered} & \begin{gathered}199784641270216\\698888192000\end{gathered} & \begin{gathered}180675787324736\\8832561971200\end{gathered} & \begin{gathered}599640352202452\\6875879014400\end{gathered} &\begin{gathered}691279186416839\\1847603077120\end{gathered}\\
%\hline
\end{array}
$$
\caption{\label{P22tau21}$\langle\tau_1(\phi_2)^k\rangle_{g,d=2-2g+k}$ for $\mathbb{P}^1_{2,2}$.}
\end{table}

\end{landscape}

\begin{landscape}
\begin{table}[ht]
\renewcommand{\arraystretch}{1.2}
$$
\begin{array}{|c|c|c|c|c|}
\hline
k&g=0&g=1&g=2&g=3\\
\Xhline{1pt}
2&0&0&0&0\\
\hline
 4 & 28 & \frac{340}{3} & \frac{1031}{24} & \frac{169}{864} \\
\hline
6&0&0&0&0\\
\hline
 8 & 1992640 & 31538360 & \frac{370089391}{3} & \frac{24266946095}{216} \\
\hline
10&0&0&0&0\\
\hline
 12 & 2047764586560 & 78504956006400 & 1001296376677905 & \frac{14545099547098120}{3} \\
\hline
%14&0&0&0&0\\
%\hline
% 16 & \begin{gathered}107781838194750\\72000\end{gathered} & \begin{gathered}795371201299570\\896000\end{gathered} & \begin{gathered}225937725097961\\57269600\end{gathered} & \frac{2686025121651702362431900}{9} \\
%\hline
%18&0&0&0&0\\
%\hline
% 20 & \begin{gathered}183018923941387\\597487424000\end{gathered} &\begin{gathered} 227140297192843\\02261967104000\end{gathered} & \begin{gathered}1183591908841482\\012973396296000\end{gathered} & \frac{96176385221840444192456732948000}{3} \\
%\hline
\end{array}
$$
\caption{\label{P22tau12}$\langle\tau_2(\phi_1)^k\rangle_{g,d=2-2g+3k/2}=\langle\tau_2(\phi_3)^k\rangle_{g,d=2-2g+3k/2}$ for $\mathbb{P}^1_{2,2}$.  By~\eqref{poly-degree-dimension} these GW invariants with odd $k$ vanish.}
\end{table}

\begin{table}[ht]
\renewcommand{\arraystretch}{1.2}
$$
\begin{array}{|c|c|c|c|c|c|}
\hline
k&g=0&g=1&g=2&g=3&g=4\\
\Xhline{1pt}
 1 & \frac{1}{4} & \frac{7}{24} & \frac{7}{5760} & 0 & 0 \\
\hline
 2 & \frac{2}{3} & \frac{11}{6} & \frac{49}{288} & 0 & 0 \\
\hline
 3 & 4 & \frac{127}{6} & \frac{595}{48} & \frac{343}{3456} & 0 \\
\hline
 4 & 40 & \frac{1072}{3} & \frac{6839}{12} & \frac{4615}{54} & \frac{2401}{41472} \\
\hline
 5 & 576 & \frac{23854}{3} & \frac{230545}{9} & \frac{729965}{48} & \frac{3113561}{5184} \\
\hline
 6 & 10976 & 219776 & \frac{7322833}{6} & \frac{47704670}{27} & \frac{157790071}{384} \\
\hline
 7 & 262144 & \frac{21783992}{3} & 62940696 & \frac{38000806025}{216} & \frac{9851843729}{81} \\
\hline
 8 & 7558272 & \frac{837523456}{3} & \frac{31750669160}{9} & 16602314176 & \frac{32332834724575}{1296} \\
\hline
 9 & 256000000 & 12244336032 & 214586106112 & \frac{13986025758950}{9} & 4253712401232 \\
\hline
 10 & 9977431552 & \frac{1810030428160}{3} & 14115273880680 & \frac{1328570203735040}{9} & \frac{71028706834703261}{108} \\
\hline
 11 & 440301256704 & \frac{98998713882496}{3} & \frac{9000091860981760}{9} & 14406010071551010 & \frac{2609747845143143936}{27} \\
\hline
 12 & 21718014715904 & 1983802353647616 & \frac{227974570522490176}{3} & \frac{39255961847179264000}{27} & \frac{27763964039632169943}{2} \\
\hline
\end{array}
$$
\caption{\label{P22tau22}$\langle\tau_2(\phi_2)^k\rangle_{g,d=2-2g+2k}$ for $\mathbb{P}^1_{2,2}$.}
\end{table}
\end{landscape}

\section{Proof of Theorem~\ref{thm-GW}}\label{conclusion}
In this section we prove Theorem~\ref{thm-GW} along the line given in~\cite{DYZ} (see also~\cite{BDY3, DYZ2})
using the matrix-resolvent method~\cite{BDY1, BDY3, DY1, FYZ}.
Most time of this section we restrict to the $m_2=1$ case. 

We first review the work from~\cite{FYZ}. 
Denote by $\mathcal{A}$ the ring of polynomials of $u_{\alpha, \, ix}$, $-1\leq \alpha\leq m_1-1$, $i\ge0$. 
Recall from~\cite{Carlet} that the bigraded Toda hierarchy with $m_2=1$ is defined by
\beq
\e\frac{\p L}{\p t^a_k} = \Bigl[\bigl(L^{\frac{a}{m_1}+k}\bigr)_+, L\Bigr], \quad a=1,\dots,m_1, \, k\ge0.
\eeq
Here $L$ is the Lax operator (cf.~\eqref{definitionL}). See~\cite{Carlet} for details about the definition.  
As in~\cite{DY1, FYZ} denote by~$\mathcal{L}$ the matrix Lax operator associated to~$L$, which is given by
\beq
\mathcal{L}:=\TT+\Lambda(\lambda)+V,
\eeq
where $\TT=e^{\e\p_x}$, $\Lambda(\lambda)=-\lambda e_{1,m_1}-\sum_{i=1}^{m_1}e_{i+1,i}$, $V=\sum_{j=1}^{m_1+1}u_{m_1-j}e_{1,j}$.
The \textit{basic matrix resolvents of~$\mathcal{L}$}, denoted $R_a(\lambda)$, $a=1,\dots, m_1$, 
are defined~\cite{FYZ} as the unique elements in $\mathcal{A}[[\e]]\otimes{\rm Mat}(l\times l,\CC((\lambda^{-1})))$ satisfying:
\begin{align}
&
\TT(R_a(\lambda))(\Lambda(\lambda)+V)-(\Lambda(\lambda)+V)R_a(\lambda)=0, \label{def-MR} \\
&{\rm Tr} \, R_a(\lambda)R_b(\lambda)=m\lambda\delta_{a+b,m_1}+m_1\lambda^2\delta_{a+b,2m_1}, 
\quad {\rm Tr} \, R_a(\lambda)=m_1\delta_{a,m_1}\lambda,\label{trace-MR}  \\
&R_a(\lambda)=\Lambda_a(\lambda)+\mbox{lower order terms with respect to }  \overline{\rm deg},\label{degree-MR}\\
&R_a(\lambda) \mbox{ is homogenous of degree } a \mbox{ with respect to }\overline{\rm deg}^e,\label{hom-MR}
\end{align}
where $a, b=1,\dots,m_1$, and 
$\Lambda_a(\lambda):=(-\Lambda(\lambda))^a$. Here the gradation $\overline{\rm deg}$ on ${\rm Mat}(l\times l,\mathbb{C}((\lambda^{-1})))$ is defined by 
assigning the degrees
\beq
\overline{\rm deg} \, \lambda=m_1,\quad \overline{\rm deg} \, e_{i,j}=i-j.
\eeq
and its extention $\overline{\rm deg}^e$ on $\mathcal{A}[[\e]]\otimes{\rm Mat}(l\times l,\mathbb{C}((\lambda^{-1})))$ is defined by further assigning
\beq
\overline{\rm deg}^e \e=0, 
\quad \overline{\rm deg}^e \p_x=0, 
\quad \overline{\rm deg}^e u_{\alpha} = m_1-\alpha ,\quad -1\leq \alpha\leq m_1-1.
\eeq
\begin{lemma}
The basic matrix resolvents satisfy 
\beq\label{power-MR}
R_a(\lambda)=R_1(\lambda)^a, \quad a=1,\dots,m_1.
\eeq
\end{lemma}
\begin{proof}
Follows from the above-mentioned uniqueness for the basic matrix resolvents.
\end{proof}

\begin{lemma}\label{lemma-power}
The basic resolvents $R_a(\lambda)$ as elements in $\mathcal{A}[[\e]]\otimes{\rm Mat}(l\times l,\CC((\lambda^{-1})))$ can 
also be uniquely determined by~\eqref{def-MR}--\eqref{degree-MR} and~\eqref{power-MR}.
\end{lemma}

\begin{proof}
When $m_1=1$, the statement was known in~\cite{DY1}.
We can assume that $m_1\ge2$.  
Suppose that $\widetilde R_a(\lambda)$, $a=1,\dots,m_1$, satisfy~\eqref{def-MR}--\eqref{degree-MR} and~\eqref{power-MR}. 
It suffices to prove that $\widetilde R_1(\lambda)=R_1(\lambda)$.
Recall that the notion of resolvent manifold~\cite{BDY3,DS, FYZ} 
ensures the existence of a matrix~$U$ such that $R_a=e^{U}\Lambda_a(\lambda)e^{-U}$ and that 
\beq
e^{-U}\widetilde R_1(\lambda)e^U=:\widetilde \Lambda_1(\lambda)\in
{\rm Span}_{\CC((\lambda^{-1}))}\{I_{m_1+1},
\Lambda_1(\lambda),\dots,\Lambda_{m_1}(\lambda)\}
\eeq
(see~\cite[Lemmas~3.3--3.4]{FYZ} for the detail).
Since $\Lambda_a(\lambda)=(-\Lambda(\lambda))^a, a=1,\dots,m_1$,
we can write
\beq\label{polytilde}
\widetilde\Lambda_1(\lambda)=
P(\Lambda_1(\lambda)),\quad 
P(y):=\sum_{a=0}^{m_1}p_a(\lambda)y^a,
\eeq
where $p_a(\lambda)\in\CC((\lambda^{-1})), a=0,\dots,m_1$.
By a direct computation we know that $\Lambda_1(\lambda)$ has 
the distinct eigenvalues $0$, $\xi_{m_1}^{j-1}\lambda^{\frac1{m_1}}$, $j=1,\dots,m_1$.
Here $\xi_{m_1}$ denotes the $m_1$th root of unity as in~\eqref{rootofunity}.
Using~\eqref{power-MR}, we find that $\widetilde \Lambda_1(\lambda)$ has the same set of eigenvalues, 
namely, the polynomial $P(y)$ introduces an $l$-permutation on $\{0, \xi_{m_1}^{0}\lambda^{\frac1{m_1}},\dots,\xi_{m_1}^{m_1-1}\lambda^{\frac1{m_1}}\}$.
Since $m_1\ge2$ and since~\eqref{degree-MR}, we find that $P(y)\equiv y$. The lemma is proved. 
\end{proof}

Let $(u_{-1}(x,\mathbf{t};\e),\dots,u_{m_1-1}(x,\mathbf{t};\e))$
 be an arbitrary solution to the bigraded Toda hierarchy, and 
  $R_a(\lambda;x,\mathbf{t};\e)$ the basic matrix resolvents $R_a(\lambda)$ evaluated at this solution. 
 Here ${\bf t}=(t^a_k)_{a=1,\dots,m_1,\,k\ge0}$.
It was shown in~\cite[Lemma 1.7]{FYZ} that there exists a function $\tau(x,\mathbf{t};\e)$, called 
the {\it tau-function of the solution $(u_{-1}(x,\mathbf{t};\e),\dots,u_{m_1-1}(x,\mathbf{t};\e))$}, satisfying 
\begin{align}
&\sum_{i,j\ge0}\frac{\e^2\frac{\partial^2\log\tau(x,\mathbf{t};\e)}{\partial t^a_i\partial t^b_j}}{\lambda^{i+1}\mu^{j+1}}=\frac{\Tr \, R_a(\lambda;x,\mathbf{t};\e)R_b(\mu;x,\mathbf{t};\e)}{(\lambda-\mu)^2}-\frac{(a\lambda+b\mu)\delta_{a+b,m_1}+m_1\lambda\mu\delta_{a+b,2m_1}}{(\lambda-\mu)^2},\label{69}\\
&\delta_{a,m_1}+\sum_{i\ge0}\frac{\e}{\lambda^{i+1}} (\TT-1) 
\biggl(\frac{\partial \log\tau(x,\mathbf{t};\e)}{\partial t^a_i}\biggr)=(R_a(\lambda;x+\e,\mathbf{t};\e))_{m_1+1,1},\label{70}\\
&\frac{\tau(x+\e,\mathbf{t};\e)\tau(x-\e,\mathbf{t};\e)}{\tau(x,\mathbf{t};\e)^2}=u_{-1}(x,\mathbf{t};\e).\label{71}
\end{align}
Here $a,b=1,\dots,m_1$.
The function 
$\tau(x,\mathbf{t};\e)$ is uniquely determined by the solution up to multiplying by the exponential of a linear function of~$x, {\bf t}$.

Before giving the proof of Theorem~\ref{thm-GW}, we do a further preparation in the next lemma. 

We denote the basic matrix resolvents of $\mathcal{L}$ evaluated at the solution corresponding to GW invariants of $\mathbb{P}^1_{m_1,m_2}$
by $\mathcal{R}^{\rm top}_a(\lambda; x,\bt;\e)$. Recall that this solution can be determined by the 
 initial data \eqref{BTH-initial1}, \eqref{BTH-initial2} at $\mathbf{t}=\mathbf{0}$. % by $\mathcal{R}_a(\lambda,x;\e)$, $a=1,\dots,m_1$.
 Here we note that the indeterminates $T^a_i$ and $t^a_i$ are related by $T^a_i=q_{a,i} t^a_i$, $a\ge1$, $i\ge0$.
 
 Similar to~\cite{BDY3, DYZ, DYZ2} let us prove the following lemma. 
\begin{lemma} \label{bispRM}
For $a=1,\dots,m_1$, we have
\beq\label{MR-TE}
\mathcal{R}^{\rm top}_a(\lambda; x,\bt={\bf 0};\e) = \e^{1-\frac a{m_1}} \lambda^{\frac{a}{m_1}}M_a\Bigl(\frac{\lambda-x}{\e},\frac{1}\e\Bigr), 
\eeq
where 
$M_a(z,s)$ are the unique formal solutions to the TDE obtained in Theorem~\ref{prop1} with $m_2=1$. 
\end{lemma}
\begin{proof}
By using the TDE~\eqref{TE}, Corollary~\ref{property}, Proposition~\ref{thm-form}, we 
see that the right-hand side of~\eqref{MR-TE} satisfies~\eqref{def-MR}, \eqref{trace-MR}, \eqref{degree-MR} and~\eqref{power-MR}.
The statement then follows from Lemma~\ref{lemma-power}.
\end{proof}
\begin{remark}
Lemma~\ref{bispRM} tells that the basic matrix resolvents $\mathcal{R}^{\rm top}_a(\lambda; x,\bt={\bf 0};\e)$ have the $M$-bispectrality, which confirms a conjecture in \cite[Section 6.1]{DYZ} for the model under consideration.
\end{remark}

\begin{proof}[Proof of Theorem~\ref{thm-GW}]
For the case when $m_2=1$, 
from~\eqref{71} we see that 
the tau-structure defined in~\cite{FYZ} and that defined in~\cite{Carlet}, excluding higher logarithmic flows, can only possibly differ by 
a sequence of additive constants. So the tau-function defined by~\eqref{69}--\eqref{71} and 
the one in~\cite{Carlet} can only possibly differ by multiplying by the 
exponential of a quadratic function in $x,\bt$.
The validity of Conjecture~\ref{cnj1} with $m_2=1$ and $k\ge3$ then 
 follows from the result of~\cite{CL}, Lemma~\ref{bispRM} and~\cite[Proposition 1.6]{FYZ}. 
 It follows from~\eqref{hom-MR} that the tau-structure defined in~\cite{FYZ}
 is homogeneous with respect to $\overline{\rm deg}^e$, and it is easy to verify that the tau-structure 
 (excluding higher logarithmic flows) from the definition of~\cite{Carlet}
 is homogeneous of the same degree (which is nonzero except for the already-matched $\log u_{-1}$) with respect to $\overline{\rm deg}^e$. 
 This, together with the fact that the two tau-structures 
 are both homogeneous of degrees~0 with respect to $\deg$, leads to the validity of Conjecture~\ref{cnj1} with $m_2=1$
 (the variable $Q$ is recovered by the degree-dimension matching).
 Here
 $\deg \e:=1$ and $\deg \p_x:=-1$.
By Corollary~\ref{cor-symm2} this validity gives the validity of Conjecture~\ref{cnj1} with $m_1=1$.
\end{proof}

\end{document}